\documentclass[letterpaper,11pt]{article}

\usepackage[width=172mm,height=224mm]{geometry}
\usepackage{amsmath}
\usepackage{amsthm}
\usepackage{mathtools}
\usepackage{txfonts}
\usepackage{url}
\usepackage{enumerate}
\usepackage{bm}

\newcommand{\classP}{\mathrm{P}}
\newcommand{\NP}{\mathrm{NP}}
\newcommand{\EXP}{\mathrm{EXP}}
\newcommand{\NEXP}{\mathrm{NEXP}}

\newcommand{\MIPstar}{\mathrm{MIP}^*}
\newcommand{\xorMIPstar}{{\oplus}\mathrm{MIP}^*}
\newcommand{\naPCP}{\mathrm{naPCP}}
\newcommand{\poly}{\mathrm{poly}}
\newcommand{\xor}{\oplus}

\newcommand{\vct}{\bm}

\DeclarePairedDelimiter{\abs}{\lvert}{\rvert}
\DeclarePairedDelimiter{\norm}{\lVert}{\rVert}
\DeclarePairedDelimiter{\bra}{\langle}{\rvert}
\DeclarePairedDelimiter{\ket}{\lvert}{\rangle}
\DeclarePairedDelimiter{\avg}{\langle}{\rangle}

\DeclareMathOperator{\tr}{tr}

\newcommand{\wcons}{w_{\mathrm{cons}}}
\newcommand{\wproof}{w_{\mathrm{sim}}}

\newcommand{\calH}{{\mathcal{H}}}
\newcommand{\calK}{{\mathcal{K}}}
\newcommand{\calM}{{\mathcal{M}}}
\newcommand{\calN}{{\mathcal{N}}}
\newcommand{\calS}{{\mathcal{S}}}
\newcommand{\RR}{{\mathbb{R}}}
\newcommand{\CC}{{\mathbb{C}}}
\newcommand{\NN}{\mathbb{N}}
\newcommand{\ZZ}{\mathbb{Z}}
\newcommand{\MS}{{\mathrm{MS}}}
\newcommand{\wc}{w_{\mathrm{c}}}
\newcommand{\wq}{w_{\mathrm{q}}}
\newcommand{\wQ}{w_{\mathrm{com}}}
\newcommand{\wns}{w_{\mathrm{ns}}}

\newcommand{\Pc}{P_{\mathrm{c}}}
\newcommand{\PQ}{P_{\mathrm{com}}}

\newcommand*{\function}[3]{{#1 \colon #2 \rightarrow #3}}

\theoremstyle{plain}
\newtheorem{theorem}{Theorem}
\newtheorem{lemma}[theorem]{Lemma}
\newtheorem{corollary}[theorem]{Corollary}
\newtheorem{claim}{Claim}
\theoremstyle{remark}
\newtheorem{remark}{Remark}

\makeatletter
\renewcommand\@maketitle{%
  \newpage
  \null
  \vskip 2em%
  \begin{center}%
  \let \footnote \thanks
    {\Large\bfseries \@title \par}%
    \vskip 1.5em%
    {\large
      \lineskip .5em%
      \begin{tabular}[t]{c}%
        \@author
      \end{tabular}\par}%
  \end{center}%
  \par
  \vskip 1em}
\makeatother

\title{
  Generalized Tsirelson Inequalities,
  Commuting-Operator Provers, \\
  and Multi-Prover Interactive Proof Systems
}

\urldef\emailTI\url{tsuyoshi@cs.mcgill.ca}
\urldef\emailHK\url{hirotada@nii.ac.jp}
\urldef\emailDP\url{dpreda@eecs.berkeley.edu}
\urldef\emailXS\url{xiaomings@tsinghua.edu.cn}
\urldef\emailAY\url{andrewcyao@tsinghua.edu.cn}

\def\and{%
  \end{tabular}%
  \hskip 0em plus .17fil%
  \begin{tabular}[t]{c}}%
\author{%
  Tsuyoshi Ito\footnotemark[1] \\
  School of Computer Science \\
  McGill University \\
  Montreal, Canada \\
  \emailTI
  \and
  Hirotada Kobayashi\footnotemark[2] \\
  Principles of Informatics Research Division \\
  National Institute of Informatics \\
  Tokyo, Japan \\
  \emailHK
  \and
  Daniel Preda \\
  Computer Science Division \\
  University of California, Berkeley \\
  USA \\
  \emailDP
  \and
  Xiaoming Sun\footnotemark[3] \\
  Center for Advanced Study \\
  Tsinghua University \\
  Beijing, China \\
  \emailXS
  \and
  Andrew C.-C. Yao\footnotemark[3] \\
  Center for Advanced Study \\
  Tsinghua University \\
  Beijing, China \\
  \emailAY}

\date{April~7, 2008}

\begin{document}
\sloppy

\makeatletter
\begingroup
  \renewcommand\thefootnote{\fnsymbol{footnote}}%
  \def\@makefnmark{\rlap{\@textsuperscript{\normalfont\@thefnmark}}}%
  \long\def\@makefntext#1{\parindent 1em\noindent
          \hb@xt@1.8em{%
              \hss\@textsuperscript{\normalfont\@thefnmark}}#1}%
  \if@twocolumn
    \ifnum \col@number=\@ne
      \@maketitle
    \else
      \twocolumn[\@maketitle]%
    \fi
  \else
    \newpage
    \global\@topnum\z@   
    \@maketitle
  \fi
  \thispagestyle{plain}\@thanks
  \begin{center}
    \large
    \@date\\
    [8mm]
  \end{center}
  \footnotetext[1]{%
    Research carried out while at the National
    Institute of Informatics.
    Supported by the Grant-in-Aid
    for Scientific Research on Priority Areas
    No.~18079014 of the Ministry of Education, Culture, Sports,
    Science and Technology of Japan.
  }
  \footnotetext[2]{%
    Supported by
    the Grant-in-Aid for Scientific Research (B)
    No.~18300002 of the Ministry of Education, Culture, Sports,
    Science and Technology of Japan.
  }
  \footnotetext[3]{%
    Supported by the National
    Natural Science Foundation of China Grant 60553001, 60603005, and
    the National Basic Research Program of China Grant
    2007CB807900, 2007CB807901.
  }
\endgroup

\setcounter{footnote}{0}%
\global\let\thanks\relax
\global\let\maketitle\relax
\global\let\@maketitle\relax
\global\let\@thanks\@empty
\global\let\@author\@empty
\global\let\@date\@empty
\global\let\@title\@empty
\global\let\title\relax
\global\let\author\relax
\global\let\date\relax
\global\let\and\relax
\makeatother

\begin{abstract}
A central question in quantum information theory and computational
complexity is how powerful nonlocal strategies are in cooperative
games with imperfect information, such as multi-prover interactive
proof systems. This paper develops a new method for proving limits
of nonlocal strategies that make use of prior entanglement among
players (or, provers, in the terminology of multi-prover interactive
proofs). Instead of proving the limits for usual isolated provers
who initially share entanglement, this paper proves the limits for
``commuting-operator provers'', who share private space, but can
apply only such operators that are commutative with any operator
applied by other provers. Obviously, these commuting-operator
provers are at least as powerful as usual isolated but
prior-entangled provers, and thus, limits in the model with
commuting-operator provers immediately give limits in the usual
model with prior-entangled provers. Using this method, we obtain an
$n$-party generalization of the Tsirelson bound for the
Clauser--Horne--Shimony--Holt inequality, for every $n$. Our bounds are
tight in the sense that, in every $n$-party case, the equality is
achievable by a usual nonlocal strategy with prior entanglement. We
also apply our method to a three-prover one-round binary interactive proof
system for $\NEXP$. Combined with the technique developed by
Kempe, Kobayashi, Matsumoto, Toner and Vidick to analyze the
soundness of the proof system, it is proved
to be $\NP$-hard to distinguish whether the entangled value of
a three-prover one-round binary-answer game is equal to one or
at most $1-1/p(n)$ for some polynomial $p$,
where $n$ is the number of questions.
This is in contrast to the two-prover one-round binary-answer case,
where the corresponding problem is efficiently decidable.
Alternatively, $\NEXP$ has a three-prover one-round binary interactive proof
system with perfect completeness and soundness $1-2^{-\poly}$.
\end{abstract}

\section{Introduction}

Nonlocality of multi-party systems
is one of the central issues in quantum information theory.
This can be naturally expressed within the framework of
\emph{nonlocal games}~\cite{CHTW04},
which are cooperative games with imperfect information.
Because of this, the nonlocality also has a strong connection
with computational complexity theory,
in particular with \emph{multi-prover interactive proof systems}~\cite{BGKW88}.
In nonlocal games,
the main interests are whether or not
the value of a game changes when parties use
nonlocal strategies that make use of prior entanglement,
and if it changes, how powerful such nonlocal strategies can be.
In multi-prover interactive proof systems,
these correspond to the questions
if dishonest but prior-entangled provers
can break the original soundness condition of the system
that is assured for any dishonest classical provers,
and if so, how much amount they can deviate from the original soundness condition.

\subsection{Our contribution}

The main contribution of this paper is
to develop a new method for proving limits
of nonlocal strategies that make use of prior entanglement among players
(or, provers, in the terminology of multi-prover interactive proofs
--- this paper uses ``player'' and ``prover'' interchangeably).
Specifically, we consider \emph{commuting-operator provers},
the notion of which was already introduced in the seminal paper by
Tsirelson~\cite{Tsirelson80}.
In contrast to usual provers for multi-prover interactive proofs,
commuting-operator provers are no longer isolated,
and share a private space corresponding to a Hilbert space $\calH$.
Initially, they have some state $\ket{\varphi} \in \calH$,
and when the $k$th prover $P_k$ receives a question $i$,
he applies some predetermined operation $A_i^{(k)}$ acting over $\calH$.
The only constraint for the provers is
that operators $A_i^{(k)}$ and $A_j^{(l)}$
of different provers $P_k$ and $P_l$
always commute for any questions $i$ and $j$.
It is obvious from this definition
that these commuting-operator provers are at least as powerful as
usual isolated but prior-entangled provers,
and thus, limits in the model with commuting-operator provers
immediately give limits in the usual model with prior-entangled provers.
Using these commuting-operator provers,
or more precisely,
making intensive use of the commutative properties of operators,
we obtain a number of intriguing results
on the limits of nonlocal strategies.

We first show a tight bound of the strategies of commuting-operator players
for the generalized ${n \times n}$ Magic Square game played by $n$ players.
This bound is naturally interpreted as
an $n$-party generalization of the Tsirelson bound
for the Clauser-Horne-Shimony-Holt (CHSH) inequality,
and thus, we essentially obtain a family of generalized Tsirelson-type inequalities,
as stated in the following theorem.

\begin{theorem}  \label{theorem:generalized-tsirelson}
  Let $X^{(i)}_j$ be $\pm1$-valued observables on $\calH$ for $0\le i\le n-1$ and $1\le j\le n$
  where $X^{(i)}_j$ and $X^{(i')}_{j'}$ commute if $i\ne i'$ ($\forall 1\leq j,j'\leq n$).
  Let $M_j=\prod_{i=0}^{n-1} X^{(i)}_j$ and $N_k=\prod_{i=0}^{n-1} X^{(i)}_{k-i}$ be observables
  for $1\le j,k\le n$,
  where the subscript $k-i$ is interpreted under modulo $n$.
  Then,
  \begin{equation}
    \sum_{j=1}^n \avg{M_j}
    +
    \sum_{k=1}^{n-1} \avg{N_k}
    -
    \avg{N_n}
    \leq
    2n \cos \frac{\pi}{2n},
    \label{eq:inequality}
  \end{equation}
  where $\avg{{\cdot}}$ denotes expected value.
\end{theorem}

\noindent
In particular, for $n=2$, our inequality is identical to the Tsirelson bound for the CHSH inequality.
For $n=3$, we have the following corollary,
which was originally proved with a different proof
in a preliminary work by a subset of the authors
(Sun, Yao and Preda~\cite{SYP07}).

\begin{corollary} \label{corollary:generalized-tsirelson-3}
  Let $X^{(i)}_j$ be $\pm1$-valued observables on $\calH$ for $1\le i,j\le 3$
  where $X^{(i)}_j$ and $X^{(i')}_{j'}$ commute if $i\ne i'$ ($\forall 1\leq j,j'\leq 3$).
  Then,
  \[
     \avg{X^{(1)}_1X^{(2)}_1X^{(3)}_1}
    +\avg{X^{(1)}_2X^{(2)}_2X^{(3)}_2}
    +\avg{X^{(1)}_3X^{(2)}_3X^{(3)}_3}
    +\avg{X^{(1)}_1X^{(2)}_3X^{(3)}_2}
    +\avg{X^{(1)}_2X^{(2)}_1X^{(3)}_3}
    -\avg{X^{(1)}_3X^{(2)}_2X^{(3)}_1}
    \le 3\sqrt3.
  \]
\end{corollary}

Theorem~\ref{theorem:generalized-tsirelson}
includes the inequalities proved by Wehner~\cite{Wehner06} as special cases
--- our proof is completely different from hers.
It is stressed that the inequalities in Theorem~\ref{theorem:generalized-tsirelson}
and Corollary~\ref{corollary:generalized-tsirelson-3}
are tight even in the usual nonlocal model with prior entanglement,
a simple proof of which is also given in this paper.

In terms of Magic Square games,
Theorem~\ref{theorem:generalized-tsirelson} implies the following.

\begin{corollary}  \label{corollary:magic-square}
For every $n\ge2$,
the maximum winning probability in the $n$-player Magic Square game
both for commuting-operator players and for usual prior-entangled players
is equal to ${(1 + \cos \frac{\pi}{2n})/2}$.
\end{corollary}

Next we prove the limits of the strategies of commuting-operator
provers for three-prover one-round interactive proof systems for
$\NP$ and $\NEXP$. The proof system makes use of
three-query non-adaptive probabilistically checkable proof (PCP)
systems with perfect completeness due to H\aa stad~\cite{Hastad01}.
Because of the commutative properties of operators each prover
applies, it is quite easy to apply the technique developed by
Kempe,~Kobayashi,~Matsumoto,~Toner,~and~Vidick~\cite{KKMTV-0704.2903v2} when
analyzing the soundness accepting probability of our system.
With this analysis, we show that it is $\NP$-hard to compute the value of
a three-player one-round \emph{binary-answer} game with entangled players,
which improves the original result in Ref.~\cite{KKMTV-0704.2903v2}
where a ternary answer from each prover was needed for the $\NP$-hardness.
In fact, we show that it is $\NP$-hard even to decide if the value of
a three-player one-round binary-answer game is one or not.
In sharp contrast to this,
the result by Cleve, H\o yer, Toner and Watrous~\cite{CHTW04}
implies that the corresponding decision problem is in $\classP$
in the case with a \emph{two-player} one-round binary-answer game.
Alternatively, we show that
any language in $\NEXP$ has a three-prover one-round \emph{binary} interactive proof system
of perfect completeness with soundness $1 - 2^{\poly}$,
whereas only languages in $\EXP$ have such proof systems in the two-prover one-round binary case.

More precisely, let $\naPCP_{c(n),s(n)}(r(n),q(n))$ be the class of languages
recognized by a probabilistically checkable proof system
with completeness and soundness acceptance probabilities $c(n)$ and $s(n)$
such that the verifier uses $r(n)$ random bits
and makes $q(n)$ non-adaptive queries,
and let $\MIPstar_{c(n),s(n)}(m,1)$ be the class of languages
recognized by a classical $m$-prover one-round interactive proof system
with entangled provers
with completeness and soundness acceptance probabilities $c(n)$ and $s(n)$.
Our main technical theorem is stated as follows.

\begin{theorem}  \label{theorem:pcp-3provers}
  $\naPCP_{1,s(n)}(r(n),3)\subseteq\MIPstar_{1,1-\varepsilon(n)}(3,1)$,
  where $\varepsilon(n)=(1/384)(1-s(n))^2\cdot2^{-2r(n)}$.
  In this interactive proof system,
  the verifier uses $r(n)+O(1)$ random bits,
  each prover answers one bit,
  and honest provers do not need to share prior entanglement.
  Moreover, the soundness of the interactive proof system holds
  also against commuting-operator provers.
\end{theorem}

By applying Theorem~\ref{theorem:pcp-3provers}
to well-known inclusions
$\NP\subseteq\bigcup_{c>0}\naPCP_{1,1-1/n^c}(c\log n,3)$
and
$\NEXP\subseteq\bigcup_{c>0}\naPCP_{1,1-2^{-cn}}(n^c,3)$,
which come
from the $\NP$-completeness of the 3SAT problem
and the $\NEXP$-completeness of the succinct version of 3SAT
(see e.g.\ Ref.~\cite{DK00}),
we obtain the following corollaries.

\begin{corollary}  \label{corollary:game-inapproximable}
  There exists a polynomially bounded function $\function{p}{\ZZ_{\ge0}}{\NN}$
  such that, given a classical three-player one-round binary-answer game
  with $n$ questions
  with entangled (or commuting-operator) players,
  it is $\NP$-hard to decide whether the value of the game is one
  or at most $1-p(n)$.
  Here a game is given as a description
  of a probability distribution over three-tuples of questions
  and a table showing whether the answers are accepted or not
  for each tuple of questions and each tuple of answers.
\end{corollary}

\begin{corollary}  \label{corollary:nexp-3provers}
  $\NP\subseteq\MIPstar_{1,1-1/\poly}(3,1)$ and $\NEXP\subseteq\MIPstar_{1,1-2^{-\poly}}(3,1)$,
  where the verifier uses $O(\log n)$ (resp.\ $\poly(n)$) random bits,
  each prover answers one bit,
  and honest provers do not need to share prior entanglement.
\end{corollary}

In contrast to Corollaries~\ref{corollary:game-inapproximable}
and \ref{corollary:nexp-3provers},
the following result in the two-prover case
is immediate from the result by Cleve, H\o yer, Toner and Watrous~%
\cite[Theorem~5.12]{CHTW04}.

\begin{theorem}  \label{theorem:2provers}
  \begin{enumerate}[(i)]
  \item
    Given a classical two-player one-round binary-answer game
    with entangled players,
    the problem of deciding whether the value of the game is
    equal to one or not is in $\classP$.
  \item
    Only languages in $\EXP$
    have two-prover one-round binary interactive proof systems
    with entangled provers of perfect completeness
    and soundness acceptance probability $1-2^{-\poly}$.
  \end{enumerate}
\end{theorem}

An important consequence of Tsirelson's theorem~\cite{Tsirelson80}
is that, using semidefinite programming,
it is easy to compute the maximum winning probability
of a so-called two-player one-round XOR game with entangled players,
which is a two-player one-round binary-answer game with entangled players
in which the result of the game only depends on the XOR of the answers from the players.
Corollary~\ref{corollary:game-inapproximable} shows
that this is not the case if we consider three players
and we drop the XOR condition of the game unless $\classP=\NP$.

\subsection{Background}

Multi-prover interactive proof systems (MIPs) were proposed by Ben-Or,
Goldwasser, Kilian and Wigderson~\cite{BGKW88}. It was proved
by Babai, Fortnow and Lund~\cite{BFL91}
that the power of MIPs is exactly equal to $\NEXP$.
Subsequently, it was shown that they still
achieve $\NEXP$ even in the most restrictive setting of two-prover
one-round interactive proof systems~\cite{FL92}. One of the main
tools when proving these claims is
the \emph{oracularization}~\cite{BGKW88,FRS94}, which forces provers to act just
like fixed proof strings.

Cleve, H\o yer, Toner and Watrous~\cite{CHTW04}
proved many examples of two-player games
where the existence of entanglement increases winning probabilities,
including the Magic Square game,
which is an example of breakage of the oracularization paradigm
under the existence of entanglement.
They also proved that
two-prover one-round XOR proof systems,
or the proof systems where each prover's answer is one bit long
and the verifier depends only on the XOR of the answers,
recognize $\NEXP$ without prior entanglement
but at most $\EXP$ with prior entanglement.

Kobayashi and Matsumoto~\cite{KM03} showed that multi-prover
interactive proof systems with provers sharing at most polynomially
many prior-entangled qubits can recognize languages only in $\NEXP$
(even if we allow quantum messages between the verifier and each
prover). On the other hand, if provers are allowed to share
arbitrary many prior-entangled qubits, very little were known about
the power of multi-prover interactive proof systems except for the
case of XOR proof systems. Very recently,
Kempe,~Kobayashi,~Matsumoto,~Toner and Vidick~\cite{KKMTV-0704.2903v2} showed
that $\NP\subseteq\MIPstar_{1,1-1/\poly}(3,1)$
and $\NEXP\subseteq\MIPstar_{1,1-2^{-\poly}}(3,1)$.
Cleve, Gavinsky and Jain~\cite{CGJ07} proved that
$\NP\subseteq\xorMIPstar_{1-\varepsilon,1/2+\varepsilon}(2,1)$,
where $\xorMIPstar_{c(n),s(n)}(2,1)$ is the class of languages
recognized by a two-prover one-round XOR interactive proof
system with entangled provers.

The only known relation between the model with commuting-operator provers
and the one with usual isolated entangled provers
is that they are equivalent in the two-prover one-round setting
that involves only finite-dimensional Hilbert spaces~\cite{Tsirelson80,Tsirelson06}.

\subsection{Organization of the paper}

Section~\ref{section:preliminaries} gives definitions
on MIP systems used in later sections.
Section~\ref{section:commuting-operator} introduces
the commuting-operator-provers model which we will use later
and states some basic facts on it.
Section~\ref{section:generalized-tsirelson} discusses the $n$-player generalization
of Tsirelson's bound based on the $n\times n$ Magic Square game.
Section~\ref{section:3sat} treats the three-prover one-round binary interactive proof system
for $\NEXP$
and compares it with the two-prover case.

\section{Preliminaries} \label{section:preliminaries}

We assume basic knowledge about quantum computation, interactive proofs and probabilistically checkable proofs.
Readers are referred to textbooks on quantum computation (e.g.\ Nielsen and Chuang~\cite{NC01})
and on computational complexity (e.g.\ Du and Ko~\cite{DK00}).
Here we review basic notions of multi-prover interactive proof systems
that are necessary to define commuting-operator model in Section~\ref{section:commuting-operator}.

A multi-prover interactive proof system can be best viewed
as a sequence of cooperative games indexed by input string.

An \emph{$m$-player cooperative one-round game} (simply an \emph{$m$-player game} in this paper)
is a pair $G=(\pi,V)$ of a probability distribution $\pi$ over $Q^m$
and a predicate $V\colon Q^m\times A^m\to\{0,1\}$, where $Q$ and $A$ are finite sets.
As a convention, we denote $V(q_1,\dots,q_m,a_1,\dots,a_m)$ by $V(a_1,\dots,a_m\mid q_1,\dots,q_m)$.
In this game, a referee decides whether the players win or lose according to a predetermined rule as follows.
The referee chooses questions $q_1,\dots,q_m$ according to the distribution $\pi$
and sends the question $q_i$ to the $i$th player.
The $i$th player sends back an answer $a_i\in A$, and the referee collects the answers $a_1,\dots,a_m$.
The players win if $V(a_1,\dots,a_m\mid q_1,\dots,q_m)=1$ and lose otherwise.
In this paper, we often refer to players as ``provers'' for better correspondence to multi-prover interactive proof systems.

A \emph{behavior} or a \emph{no-signaling strategy} for $G$
is a function $S\colon Q^m\times A^m\to[0,1]$
with normalization and no-signaling conditions.
Like $V$, we denote $S(q_1,\dots,q_m,a_1,\dots,a_m)$ by $S(a_1,\dots,a_m\mid q_1,\dots,q_m)$,
and it corresponds to the probability with which the $m$ players answer $a_1,\dots,a_m$
under the condition that the questions sent to the players are $q_1,\dots,q_m$.
The \emph{normalization condition} requires
that for all $q_1,\dots,q_m\in Q$, $\sum_{a_1,\dots,a_m\in A}S(a_1,\dots,a_m\mid q_1,\dots,q_m)=1$.
The \emph{no-signaling condition} requires
that for any $1\le i\le m$, any $q_1,\dots,q_{i-1},q_{i+1},\dots,q_m\in Q$
and any $a_1,\dots,a_{i-1},a_{i+1},\dots,a_m\in A$,
the sum $\sum_{a_i\in A}S(a_1,\dots,a_m\mid q_1,\dots,q_m)$ does not depend on the choice of $q_i\in Q$.
The \emph{winning probability} $w(S)$ of the strategy $S$ is given by
\[
  w(S)=\sum_{q_1,\dots,q_m\in Q} \pi(q_1,\dots,q_m)
  \sum_{a_1,\dots,a_m\in A} S(a_1,\dots,a_m\mid q_1,\dots,q_m) V(a_1,\dots,a_m\mid q_1,\dots,q_m).
\]

A behavior is said to be \emph{classical} (resp.\ \emph{entangled})
if it is realized by a classical (resp.\ entangled) strategy.
In a \emph{classical (resp.\ entangled) strategy},
$m$ computationally unlimited players share a random source (resp.\ a quantum state),
and each of them decides his/her answer according to his/her question and the shared random source (resp.\ state).
It is well-known that for any classical strategy,
there exists an equivalent classical strategy without shared random source.
Also for any entangled strategy,
there exists an equivalent entangled strategy
where the players share a pure state and their measurements are projective.

The \emph{classical (resp.\ entangled, no-signaling) value} of $G$, denoted by $\wc(G)$ (resp.\ $\wq(G)$, $\wns(G)$),
is the supremum of the winning probabilities over all classical (resp.\ entangled, no-signaling) behaviors for $G$.
Clearly we have $0\le\wc(G)\le\wq(G)\le\wns(G)\le1$.
The classical and no-signaling values of $G$ can be attained for all games $G$,
but it is not known whether the entangled value of $G$ can be attained for all games $G$.

An $m$-prover one-round interactive proof system
is a pair $(M_\pi,M_V)$ of two Turing machines.
A probabilistic Turing machine $M_\pi$ is given an input string $x$
and outputs $m$ questions $q_1,\dots,q_m$.
A deterministic Turing machine $M_V$ is given an input $x$
and $2m$ strings $q_1,\dots,q_m,a_1,\dots,a_m$,
and outputs $0$ or $1$.
Both $M_\pi$ and $M_V$ must run in time polynomial in $\abs{x}$.
This system naturally defines an $m$-player game $G_x$ for each input string $x$.

Let $c,s\colon\ZZ_{\ge0}\to[0,1]$.
An $m$-prover one-round interactive proof system
is said to have \emph{completeness acceptance probability $c(n)$
for a language $L$ for classical (resp.\ entangled) provers}
when $\wc(G_x)\ge c(\abs{x})$ (resp.\ $\wq(G_x)\ge c(\abs{x})$) for all $x\in L$.
Similarly, it is said to have \emph{soundness acceptance probability $s(n)$
for a language $L$ for classical (resp.\ entangled) provers}
when $\wc(G_x)\le s(\abs{x})$ (resp.\ $\wq(G_x)\le s(\abs{x})$) for all $x\notin L$.

Let $\MIPstar_{c(n),s(n)}(m,1)$ denote the class of languages
having $m$-prover one-round interactive proof systems
with completeness and soundness acceptance probabilities $c(n)$ and $s(n)$
for entangled provers.

Let $\naPCP_{c(n),s(n)}(r(n),q(n))$ denote
the class of languages having PCP systems
with completeness and soundness acceptance probabilities $c(n)$ and $s(n)$
where the verifier reads $q(n)$ bits in a proof non-adaptively
using $r(n)$ random bits.

H\aa stad~\cite{Hastad01} gave the following characterizations of $\NP$ and $\NEXP$.

\begin{theorem}[H\aa stad~\cite{Hastad01}]  \label{theorem:hastad}
  For any constant $3/4<s<1$, $\NP=\bigcup_{c>0}\naPCP_{1,s}(c\log n,3)$ and $\NEXP=\bigcup_{p\in\poly}\naPCP_{1,s}(p,3)$.
\end{theorem}

It is noted that
applying Theorem~\ref{theorem:pcp-3provers} to the PCP systems
in Theorem~\ref{theorem:hastad} gives a slightly better soundness in
Corollaries~\ref{corollary:game-inapproximable} and
\ref{corollary:nexp-3provers}
(but polynomials remain polynomials).
This is not a significant improvement by itself,
but if the soundness bound
in Theorem~\ref{theorem:pcp-3provers} can be improved,
then applying it to Theorem~\ref{theorem:hastad} will become necessary
to obtain a better result on entangled provers.

\section{Commuting-operator provers}  \label{section:commuting-operator}

\subsection{Definition and basic properties}

Here we define a class of strategies called commuting-operator strategies,
which are a generalization of entangled strategies.
All the upper bounds of the entangled values of games proved in this paper
are actually valid even for this class.
A \emph{commuting-operator strategy} is a tuple $(\calH,\rho,\calM^{(i)}_q)$
of a Hilbert space $\calH$, a quantum state $\rho$ in $\calH$,
and a family of POVMs $\calM^{(i)}_q=(M^{(i)}_{q,a})_{a\in A}$ on the whole space $\calH$
for $1\le i\le m$, $q\in Q$
such that $M^{(i)}_{q,a}$ and $M^{(i')}_{q',a'}$ commute whenever $i\ne i'$:
$[M^{(i)}_{q,a},M^{(i')}_{q',a'}]=M^{(i)}_{q,a}M^{(i')}_{q',a'}-M^{(i')}_{q',a'}M^{(i)}_{q,a}=0$.
In this strategy,
$m$ players share a quantum state $\rho$,
and player~$i$ measures the state $\rho$ with $\calM^{(i)}_{q_i}$
depending on the query $q_i$ sent to him/her.
Then the joint probability of the answers $a_1,\dots,a_m$
under the condition that the questions are $q_1,\dots,q_m$ is given by
$S(a_1,\dots,a_m\mid q_1,\dots,q_m)=\tr\rho M^{(1)}_{q_1,a_1}\dotsm M^{(m)}_{q_m,a_m}$.
Such a behavior $S$ induced by a commuting-operator strategy
is called a \emph{commuting-operator behavior},
and the \emph{commuting-operator value} $\wQ(G)$ of a game $G$
is the supremum of the winning probabilities over all commuting-operator behaviors for $G$.

An entangled strategy in the usual sense with Hilbert spaces $\calH_1,\dots,\calH_m$
is a special case of commuting-operator strategies with Hilbert spaces $\calH=\calH_1\otimes\dots\otimes\calH_m$,
since for $i\ne i'$, POVMs on $\calH_i$ and POVMs on $\calH_{i'}$ commute element-wise
when they are viewed as POVMs on $\calH$.
This implies that $0\le\wc(G)\le\wq(G)\le\wQ(G)\le\wns(G)\le1$.

For the special cases of two-player binary-answer games where the
referee decides the result of the game depending only on the queries
and the XOR of the answers from the two players, the optimal
strategy for entangled players and the maximum acceptance
probability is given by optimizing certain inner products among
vectors~\cite{Tsirelson80}, and the entangled value of the game can
be computed efficiently by using semidefinite programming.
Tsirelson~\cite{Tsirelson80} also proved that this value does not
change if we replace the entangled players by commuting-operator
players. Tsirelson~\cite{Tsirelson06} generalized the equivalence of
the two models to the case of two players where the dimension of the
quantum state shared by the players is finite. However, it is not
known whether this equivalence holds for general two-player
binary-answer games.

If the outcomes of measurements are real numbers,
then the expected values of the product of the outcomes of $\calM^{(i)}_{q_i}$ for $i\in P\subseteq\{1,\dots,m\}$
is $\tr\rho\prod_{i\in P}X^{(i)}_{q_i}$ with observables $X^{(i)}=\sum_{a\in A}aM^{(i)}_{q,a}$.

The following simple observation relates the commutativity of observables
and unentangled players.

\begin{lemma} \label{lemma:commute-classical}
  If there is a commuting-operator strategy in a game $G$
  with acceptance probability $w$
  where all POVM operators $M^{(i)}_{q,a}$ commute,
  then $\wc(G)\ge w$.
\end{lemma}

\begin{proof}
  Intuitively, the lemma holds since one can measure
  all the POVMs $\calM^{(i)}_q$ simultaneously because of commutativity.
  Details follow.

  Let $a^{(i)}_q\in A$ for $1\le i\le m$ and $q\in Q$,
  and let
  $\vct{a}=(a^{(1)}_1,\dots,a^{(1)}_{\abs{Q}},a^{(2)}_1,\dots,a^{(2)}_{\abs{Q}},\dots,
            a^{(m)}_1,\dots,a^{(m)}_{\abs{Q}})$.
  We define a linear operator
  \[
    M(\vct{a})=\prod_{i=1}^m\prod_{q\in Q}M^{(i)}_{q,a^{(i)}_q}.
  \]
  By commutativity of the observables,
  $M(\vct{a})$ is Hermitian and nonnegative definite for any $\vct{a}$,
  and $\sum_{\vct{a}}M(\vct{a})=I$.

  We construct a classical strategy
  with acceptance probability $w$.
  The players share
  $a^{(1)}_1,\dots,a^{(1)}_{\abs{Q}},\allowbreak\dots,\allowbreak
   a^{(m)}_1,\dots,a^{(m)}_{\abs{Q}}\in A$
  with probability
  $\bra\psi M(\vct{a})\ket\psi$.
  The $i$th player answers $a^{(i)}_q$ when asked query $q$.
  By simple calculation, the probability distribution of the answers
  conditioned on arbitrary set of $m$ queries in the classical strategy
  is exactly equal to that in the original commuting-operator strategy.
\end{proof}

Like entangled strategies, for any commuting-operator strategy,
there exists an equivalent commuting-operator strategy
with a pure shared quantum state and projection-valued measures (PVMs).

\subsection{Symmetrization}

Here we prove that we can assume the players' optimal strategy is symmetric
under any permutations of the players.
A precise definition of the symmetry of a commuting-operator strategy follows.

Let $G=(\pi,V)$ be an $m$-player game.
$G$ is said to be \emph{symmetric}
if the following conditions are satisfied.
\begin{enumerate}[(i)]
\item
  $\pi$ is symmetric:
  $\pi(q_{\sigma(1)},\dots,q_{\sigma(m)})=\pi(q_1,\dots,q_m)$
  for any permutation $\sigma\in\calS_m$.
\item
  $V$ is symmetric under permutations of players:
  $V(a_{\sigma(1)},\dots,a_{\sigma(m)}\mid q_{\sigma(1)},\dots,q_{\sigma(m)})=V(a_1,\dots,a_m\mid q_1,\dots,q_m)$
  for any permutation $\sigma\in\calS_m$.
\end{enumerate}

Now we define the symmetry of a commuting-operator strategy. Let $\calH$
be the Hilbert space shared by the players, let $\ket\Psi\in\calH$ be
the state shared by the players, and let
$\calM^{(i)}_q=(M^{(i)}_{q,a})_{a\in A}$ be the $A$-valued PVM
measured by the player~$i$ when asked the question~$q$. The
strategy is \emph{symmetric} if there exists a unitary representation $\Phi$
of the symmetric group $\calS_m$ in $\calH$
such that $\Phi(\sigma)\ket\Psi=\ket\Psi$ and
$\Phi(\sigma^{-1})M^{(\sigma(i))}_{q,a}\Phi(\sigma)\ket\varphi=M^{(i)}_{q,a}\ket\varphi$
for any permutation $\sigma\in\calS_m$ and any state
$\ket\varphi\in\calH$.

This definition is a natural extension of the usual definition of symmetric entangled strategy
in the following sense:
consider an entangled strategy on a Hilbert space $\calH=\calK^{\otimes m}$, that is,
$\ket\Psi\in\calK^{\otimes n}$ is a state shared by the players and
$M^{(i)}_{q,a}=I\otimes\dots\otimes I\otimes M^{\prime(i)}_{q,a}\otimes I\otimes\dots\otimes I$
only acts on the $i$th tensor factor of $\calH$.
This strategy is symmetric as a commuting-operator strategy
with respect to the representation $\Phi$ of $\calS_m$ in $\calH$ defined by
$\Phi(\sigma)(\ket{\varphi_1}\otimes\dots\otimes\ket{\varphi_m})=\ket{\varphi_{\sigma^{-1}(1)}}\otimes\dots\otimes\ket{\varphi_{\sigma^{-1}(m)}}$
if and only if $M^{\prime(1)}_q=\dots=M^{\prime(m)}_q$ for all $q\in Q$.

\begin{lemma}  \label{lemma:symmetrize}
  In an $m$-player one-round symmetric game,
  if there exists a commuting-operator strategy achieving winning probability $p$,
  then there also exists a symmetric commuting-operator strategy achieving the same winning probability $p$.
\end{lemma}

\begin{proof}
  The lemma can be proved
  by constructing a symmetric strategy
  by averaging over all the permutations on provers.
  Detail follow.

  Let $(\calH,\ket\Psi,\calM^{(i)}_q)$ be a (not necessarily symmetric) commuting-operator strategy
  achieving acceptance probability $p$.
  Note that for any permutation $\tau\in\calS_m$,
  the strategy $(\calH,\ket\Psi,\calM^{(\tau(i))}_q)$ also achieves the same probability $p$
  because of the symmetry of the game.

  We construct a symmetric strategy $(\calK,\ket{\Psi'},\calN^{(i)}_q)$
  from the strategy $(\calH,\ket\Psi,\calM^{(i)}_q)$.
  Let $\calK=\calH\otimes\CC^{m!}$.
  We regard $\{\ket\tau\mid\tau\in\calS_m\}$ as an orthonormal basis of $\CC^{m!}$.
  We define a unitary representation $\Phi$
  of the symmetric group $\calS_m$ in $\calK$ as
  $\Phi(\sigma)(\ket\varphi\otimes\ket\tau)=\ket\varphi\otimes\ket{\tau\sigma^{-1}}$.
  Now we define $\ket{\Psi'}\in\calK$ by
  \[
    \ket{\Psi'}=\ket\Psi\otimes\frac{1}{\sqrt{m!}}\sum_{\tau\in\calS_m}\ket\tau.
  \]
  The player~$i$ in the constructed symmetric strategy measures the $\CC^{m!}$-part of the state,
  and acts just like the player~$\tau(i)$ in the original strategy:
  \[
    N^{(i)}_{q,a}=\sum_{\tau\in\calS_m} M^{(\tau(i))}_{q,a} \otimes \ket\tau\bra\tau.
  \]
  This strategy is a commuting-operator strategy since, for $i\ne i'$,
  \[
    \left[N^{(i)}_{q,a},N^{(i')}_{q',a'}\right]
    =\sum_{\tau\in\calS_m} \left[M^{(\tau(i))}_{q,a},M^{(\tau(i'))}_{q',a'}\right] \otimes \ket\tau\bra\tau=0.
  \]
  The symmetry of the strategy is verified as follows:
  \[
    \Phi(\sigma)\ket{\Psi'}=\ket\Psi\otimes\frac{1}{\sqrt{m!}}\sum_{\tau\in\calS_m}\ket{\tau\sigma^{-1}}=\ket{\Psi'}
  \]
  and
  \begin{align*}
    \Phi(\sigma^{-1})N^{(\sigma(i))}_{q,a}\Phi(\sigma)(\ket\varphi\otimes\ket\tau)
    &=\Phi(\sigma^{-1})N^{(\sigma(i))}_{q,a}\left(\ket\varphi\otimes\ket{\tau\sigma^{-1}}\right) \\
    &=\Phi(\sigma^{-1})\left(M^{(\tau(i))}_{q,a}\ket\varphi\otimes\ket{\tau\sigma^{-1}}\right) \\
    &=M^{(\tau(i))}_{q,a}\ket\varphi\otimes\ket\tau \\
    &=N^{(i)}_{q,a}(\ket\varphi\otimes\ket\tau).
  \end{align*}
  In the constructed strategy,
  if measurement of the $\CC^{m!}$-part of the shared state results in $\tau\in\calS_m$,
  the players just follow the strategy $(\calH,\ket\Psi,\calM^{(\tau(i))}_q)$,
  and therefore the strategy achieves winning probability $p$.
\end{proof}

\section{\boldmath $n$-party generalization of Tsirelson's bound based on $n\times n$ Magic Square}
  \label{section:generalized-tsirelson}

\subsection{Definitions and basic facts}

We define an $n$-player game for the $n\times n$ Magic Square as follows.
Consider an $n\times n$ matrix with $\{0,1\}$-entries not known to the referee.
The referee chooses one row or column randomly and uniformly.
Then he assigns the $n$ cells on the chosen row or column to the $n$ players
one-to-one randomly and uniformly,
and queries the content of each cell to the corresponding player.
Every player answers either $0$ or $1$.
The players win if and only if the sum of the $n$ answers is even,
except that, when the referee chose the column $n$,
the players win if and only if the sum of the $n$ answers is odd.
We call this game the \emph{$n$-player Magic Square game} and denote $\MS_n$.

We consider a variant of this game.
Let $L=(L_{jk})$ be a Latin square of order $n$.
That is, $L_{jk}\in\{1,\dots,n\}$
and every row or column contains $1,\dots,n$ exactly once.
We define the $n$-player Magic Square game with assignment $L$, denoted $\MS_n(L)$, as follows.
The referee chooses one row or column randomly and uniformly.
Then he queries the contents of the $n$ cells on the chosen row or column
to the $n$ players, but this time he assigns the cells to the players according to $L$:
the referee asks the $L_{jk}$-th player the content of the cell at row~$j$, column~$k$.
The rest is the same.

It is easy to verify that $\wc(\MS_n)=\wc(\MS_n(L))=1-1/(2n)$ for any Latin squares $L$,
and this classical bound corresponds to a sequence of Bell inequalities.
The Bell inequality corresponding to the two-player Magic Square game with an assignment
is known as the Clauser--Horne--Shimony--Holt (CHSH) inequality~\cite{CHSH69},
and the maximum winning probability $\wq(\MS_2(L))=\wQ(\MS_2(L))=(2+\sqrt2)/4\approx0.85$
for entangled players and even commuting-operator players
follows from the quantum version of the CHSH inequality called Tsirelson's bound~\cite{Tsirelson80}.

The following theorem states that an upper bound for the value of the game $\MS_n(L)$
is also valid for $\MS_n$.

\begin{theorem}  \label{theorem:ms-assignemt}
  For any Latin square $L$ of order $n$,
  $\wq(\MS_n)\le\wq(\MS_n(L))$ and $\wQ(\MS_n)\le\wQ(\MS_n(L))$.
\end{theorem}

\begin{proof}
  First we prove that $\wq(\MS_n)\le\wq(\MS_n(L))$.
  Consider an arbitrary entangled strategy $S$ in the game $\MS_n$.
  We construct an entangled strategy $S'$ in the game $\MS_n(L)$
  with the same winning probability as $S$.

  Let $\ket\varphi\in\calH_1\otimes\dots\otimes\calH_n$
  be the state shared by the players in $S$.
  Without loss of generality, we assume that $\calH_1=\dots=\calH_n$.
  In $S'$, the players share the state
  \[
    \ket{\varphi'}=\frac{1}{\sqrt{n!}}\sum_{\sigma\in\calS_n}
      U_\sigma\ket\varphi\otimes\ket{\sigma(1)}\otimes\dots\otimes\ket{\sigma(n)}
      \in(\calH_1\otimes\dots\otimes\calH_n)\otimes(\CC^n)^{\otimes n}
      \cong(\calH_1\otimes\CC^n)\otimes\dotsm\otimes(\calH_n\otimes\CC^n),
  \]
  where $\calS_n$ is the symmetric group on $\{1,\dots,n\}$
  and $U_\sigma$ is the unitary operator on $\calH_1\otimes\dots\otimes\calH_n$ defined by
  $U_\sigma(\ket{\varphi_1}\otimes\dots\otimes\ket{\varphi_n})
   =\ket{\varphi_{\sigma(1)}}\otimes\dots\otimes\ket{\varphi_{\sigma(n)}}$.
  Every player $i$ holds the part of $\ket{\varphi'}$
  corresponding to the space $\calH_i\otimes\CC^n$.
  When asked the content of the cell at row~$j$, column~$k$,
  the player $i=L_{jk}$ measures the $\CC^n$-part of $\ket{\varphi'}$ in the computational basis
  to obtain the value of $\sigma(i)$,
  and acts like the player $\sigma(i)$ in $S$.
  This achieves the same winning probability as $S$.

  The inequality $\wQ(\MS_n)\le\wQ(\MS_n(L))$ can be proved similarly.
  Let $S$ be a commuting-operator strategy in $\MS_n$.
  Let $\ket\varphi\in\calH$ be the state shared by the players in $S$,
  and $\calM^{(i)}_{jk}=(M^{(i)}_{jk,a})_{a\in\{0,1\}}$ be the POVM measured by player~$i$
  when he is asked the content of the cell at row~$j$, column~$k$.
  Now we consider $\CC^{n!}$ as a Hilbert space spanned by an orthonormal basis
  $\{\ket\sigma\mid\sigma\in\calS_n\}$.
  In a strategy $S'$ for $\MS_n(L)$, the commuting-operator players share the state
  \[
    \ket\varphi\otimes\frac{1}{\sqrt{n!}}\sum_{\sigma\in\calS_n}\ket\sigma
      \in\calH\otimes\CC^{n!}.
  \]
  When asked the content of the cell at row~$j$, column~$k$,
  the player $i=L_{jk}$ measures $\ket{\varphi'}$ according to the POVM
  \[
    N^{(i)}_{jk,a}=\sum_{\sigma\in\calS_n}M_{jk,a}^{(\sigma(i))}\otimes\ket\sigma\bra\sigma.
  \]
  Note that if $L_{jk}=i\ne i'=L_{j'k'}$, then $N^{(i)}_{jk,a}$ and $N^{(i')}_{j'k',a'}$ commute
  as required in the commuting-operator model since
  \[
    \left[N^{(i)}_{jk,a},N^{(i')}_{j'k',a'}\right]
    =\sum_{\sigma\in\calS_n}\left[M_{jk,a}^{(\sigma(i))},M_{j'k',a'}^{(\sigma(i'))}\right]\otimes\ket\sigma\bra\sigma=0.
    \qedhere
  \]
\end{proof}

\subsection{A strategy for entangled players}

\begin{theorem}  \label{theorem:strategy}
  There exists an entangled strategy
  in the $n$-player Magic Square game
  with winning probability $(1+\cos(\pi/(2n)))/2$.
  That is, $\wq(\MS_n)\ge(1+\cos(\pi/(2n)))/2$.
\end{theorem}

We define an $n$-qubit pure state
$\ket{\varphi_n}\in(\CC^2)^{\otimes n}$ as
\[
  \vert\varphi_n\rangle = \frac{1}{2^{(n-1)/2}}
  \Bigl( \sum_{\substack{x\in\{0,1\}^n \\ W(x)\equiv0 \bmod 4}}\ket{x}
        -\sum_{\substack{x\in\{0,1\}^n \\ W(x)\equiv2 \bmod 4}}\ket{x}\Bigr),
\]
where $W(x)$ is the number of $1$'s in $x\in\{0,1\}^n$.

We denote by $Z_\theta$ the $\pm1$-valued observable
represented by the $2\times2$ Hermitian matrix
\begin{align*}
  Z_\theta &= \begin{pmatrix}
    \cos\theta &  \sin\theta \\
    \sin\theta & -\cos\theta
  \end{pmatrix} \\
  &= \begin{pmatrix}
    \cos(\theta/2) & -\sin(\theta/2) \\
    \sin(\theta/2) &  \cos(\theta/2)
  \end{pmatrix}\begin{pmatrix}
    1 &  0 \\
    0 & -1
  \end{pmatrix}\begin{pmatrix}
     \cos(\theta/2) & \sin(\theta/2) \\
    -\sin(\theta/2) & \cos(\theta/2)
  \end{pmatrix}.
\end{align*}

The $n$ players share the $n$-qubit state $\ket{\varphi_n}$,
one qubit for each player.
When asked the content of the cell at row~$j$, column~$k$,
the player measures the observable $Z_{\theta_{jk}}$, where
\[
  \theta_{jk} = \begin{cases}
    0             & \text{if $1\le j,k\le n-1$,} \\
    \pi/(2n)      & \text{if $1\le j\le n-1$, $k=n$,} \\
    -\pi/(2n)     & \text{if $j=n$, $1\le k\le n-1$,} \\
    \pi/2         & \text{if $j=k=n$,}
  \end{cases}
\]
and answers $0$ (resp.\ $1$) if the measured value is $+1$ (resp.\ $-1$).

To prove the players win with probability $(1+\cos(\pi/(2n)))/2$,
we prepare the following lemma.

\begin{lemma}  \label{lemma:state}
  Let $n\ge1$ and $\theta_1,\dots,\theta_n\in\RR$,
  and let $\ket{\varphi_n}$ and $Z_\theta$ as defined above.
  Let $M=Z_{\theta_1}\otimes\dots\otimes Z_{\theta_n}$.
  Then,
  \[
    \langle\varphi_n\rvert M \ket{\varphi_n} = \cos(\theta_1+\dots+\theta_n).
  \]
\end{lemma}

\begin{proof}
  Let
  \[
    \vert\varphi'_n\rangle = \frac{1}{2^{(n-1)/2}}
    \Bigl( \sum_{\substack{x\in\{0,1\}^n \\ W(x)\equiv1\bmod4}}\ket{x}
          -\sum_{\substack{x\in\{0,1\}^n \\ W(x)\equiv3\bmod4}}\ket{x}\Bigr).
  \]
  We actually prove the following stronger statement:
  \begin{align*}
    \langle\varphi_n\rvert M \ket{\varphi_n}
     = -\langle\varphi'_n\rvert M \ket{\varphi'_n}
    &= \cos(\theta_1+\dots+\theta_n), \\
    \langle\varphi_n\rvert M \ket{\varphi'_n}
     = \langle\varphi'_n\rvert M \ket{\varphi_n}
    &= \sin(\theta_1+\dots+\theta_n).
  \end{align*}

  The proof is by induction on $n$.
  The case $n=1$ holds by the definition of $Z_{\theta_1}$.
  If $n>1$, note that
  \begin{align*}
    \ket{\varphi_n} &= \frac{1}{\sqrt2}(\ket{\varphi_{n-1}}\otimes\ket0-\ket{\varphi'_{n-1}}\otimes\ket1), \\
    \ket{\varphi'_n}&= \frac{1}{\sqrt2}(\ket{\varphi'_{n-1}}\otimes\ket0+\ket{\varphi_{n-1}}\otimes\ket1).
  \end{align*}
  Let $N=Z_{\theta_1}\otimes\dots\otimes Z_{\theta_{n-1}}$.
  Then,
  \begin{align*}
    \langle\varphi_n\rvert M \ket{\varphi_n}
    &= \frac12 \bigl(\bra{\varphi_{n-1}} N \ket{\varphi_{n-1}}\bra0 Z_{\theta_n}\ket0
       + \bra{\varphi'_{n-1}} N \ket{\varphi'_{n-1}}\bra1 Z_{\theta_n}\ket1 \\
    & \hphantom{=\frac12\bigl(}
       - \bra{\varphi_{n-1}} N \ket{\varphi'_{n-1}}\bra0 Z_{\theta_n}\ket1
       - \bra{\varphi'_{n-1}} N \ket{\varphi_{n-1}}\bra1 Z_{\theta_n}\ket0\bigr) \\
    &= \cos(\theta_1+\dots+\theta_{n-1})\cos\theta_n-\sin(\theta_1+\dots+\theta_{n-1})\sin\theta_n \\
    &= \cos(\theta_1+\dots+\theta_{n-1}+\theta_n).
  \end{align*}
  The other three equalities are proved similarly.
\end{proof}

It is easy to verify that $\sum_k\theta_{jk}=\pi/(2n)$ for every row~$j$.
Similarly, $\sum_j\theta_{jk}=-\pi/(2n)$ for every $k\ne n$,
and $\sum_j\theta_{jn}=\pi-\pi/(2n)$.
By Lemma~\ref{lemma:state}, the expected
value of the product of the $n$ measurement results is $\cos(\pi/(2n))$, except
that, when the referee chose the column~$n$, the expected
value of the product is $\cos(\pi-\pi/(2n))=-\cos(\pi/(2n))$. This
means that the players win with probability
$(1+\cos\frac{\pi}{2n})/2$ for every query.

\subsection{Optimality of the strategy}

We prove Theorem~\ref{theorem:generalized-tsirelson} and
Corollary~\ref{corollary:magic-square} in this section. We use the following
lemma to prove Theorem~\ref{theorem:generalized-tsirelson}.

\begin{lemma}  \label{lemma}
  Let $\calH$ be a Hilbert space, $\ket\varphi\in\calH$ be a unit vector,
  and $A,B$ be unitary operators on $\calH$.
  (We do not assume that $A$ and $B$ commute.)
  Let $\alpha=\bra\varphi A \ket\varphi$
  and $\beta=\bra\varphi B \ket\varphi$.
  Then
  $\abs[\big]{{\bra\varphi AB \ket\varphi-\alpha\beta}}
   \le \sqrt{1-\abs{\alpha}^2}\sqrt{1-\abs{\beta}^2}$.
\end{lemma}

\begin{proof}
  If $\abs\beta=1$, then $B\ket\varphi=\beta\ket\varphi$
  and the statement is trivial.
  In the rest of the proof, we assume that $\abs\beta<1$.

  Let
  \[
    \ket\psi=\frac{B\ket\varphi-\beta\ket\varphi}{\sqrt{1-\abs{\beta}^2}}.
  \]
  Then $\langle\varphi|\psi\rangle=0$ and $\langle\psi|\psi\rangle=1$.
  It follows that
  $\bra\varphi AB \ket\varphi
   =\bra\varphi A (\beta\ket\varphi+\sqrt{1-\abs{\beta}^2}\,\ket\psi)
   =\alpha\beta+\bra\varphi A \ket\psi \sqrt{1-\abs{\beta}^2}$.
  Let $\ket\xi=A^* \ket\varphi$.
  Since $\langle\varphi|\psi\rangle=0$, we have
  $\abs[\big]{{\langle\xi|\varphi\rangle}}^2+\abs[\big]{{\langle\xi|\psi\rangle}}^2\le1$.
  Note that
  $\langle\xi|\varphi\rangle=\bra\varphi A \ket\varphi=\alpha$.
  It follows that
  $\abs[\big]{{\bra\varphi A \ket\psi}}^2
   = \abs[\big]{{\langle\xi|\psi\rangle}}^2 \le 1-\lvert\alpha\rvert^2$.
  Hence
  $
    \abs[\big]{{\bra\varphi AB \ket\varphi-\alpha\beta}}^2
    =\abs[\big]{{\bra\varphi A \ket\psi}}^2 (1-\abs{\beta}^2)
    \le (1-\abs{\alpha}^2)(1-\abs{\beta}^2)
  $.
\end{proof}

\begin{corollary}  \label{corollary:1}
  Let $\calH$, $\ket\varphi$, $A$, $B$, $\alpha$ and $\beta$ be as defined in Lemma~\ref{lemma}.
  Suppose $\alpha\in\RR$, $\alpha=\cos\theta$, $\Re\beta=\cos\theta'$ with $0\le\theta,\theta'\le\pi$,
  where $\Re$ denotes the real part.
  Then
  $\cos(\theta+\theta')\le\Re\bra\varphi AB \ket\varphi\le\cos(\theta-\theta')$.
\end{corollary}

\begin{proof}
  By Lemma~\ref{lemma},
  \begin{align*}
    \abs[\big]{{\Re\bra\varphi AB \ket\varphi-\alpha\Re(\beta)}}
    &=\abs[\big]{{\Re(\bra\varphi AB \ket\varphi-\alpha\beta)}} \\
    &\le\abs[\big]{{\bra\varphi AB \ket\varphi-\alpha\beta}} \\
    &\le \sqrt{1-\alpha^2}\sqrt{1-\abs{\beta}^2} \\
    &\le \sqrt{1-\alpha^2}\sqrt{1-(\Re\beta)^2},
  \end{align*}
  which implies
  \[
       \alpha\Re(\beta)-\sqrt{1-\alpha^2}\sqrt{1-(\Re\beta)^2}
    \le\Re\bra\varphi AB \ket\varphi
    \le\alpha\Re(\beta)+\sqrt{1-\alpha^2}\sqrt{1-(\Re\beta)^2}.
  \]
  The statement follows from the facts
  that $\alpha=\cos\theta$, $\Re\beta=\cos\theta'$ and $\sin\theta,\sin\theta'\ge0$.
\end{proof}

\begin{corollary}  \label{corollary:2}
  Let $\ket\varphi$ be a unit vector in a Hilbert space $\calH$,
  let $A_1,\dots,A_n$ be Hermitian operators on $\calH$ with $A_i^2=I$,
  and let $\bra\varphi A_i \ket\varphi=\cos\theta_i$
  with $0\le\theta_i\le\pi$.
  If $\theta_1+\dots+\theta_n<\pi$,
  then $\Re\bra\varphi A_1\dotsm A_n \ket\varphi\ge\cos(\theta_1+\dots+\theta_n)>-1$.
\end{corollary}

\begin{proof}
  Use Corollary~\ref{corollary:1} repeatedly.
\end{proof}

\begin{proof}[Proof of Theorem~\ref{theorem:generalized-tsirelson}]
For notational convenience, the index $j$ in $X^{(i)}_j$ is interpreted in modulo $n$.
Let $\ket\varphi$ be the quantum state shared by the $n$ parties,
and $Z=\sum_{j=1}^n M_j+\sum_{k=1}^{n-1} N_k-N_n$.
We prove $\avg{Z}=\bra\varphi Z\ket\varphi\le2n\cos(\pi/(2n))$.

Let $P=\prod_{j=1}^n M_jN_{n+1-j}=M_1N_nM_2N_{n-1}\dotsm M_nN_1$.
We prove that $P=I$.
For $i=0,\dots,n-1$, let
\[
  P_i=\prod_{j=1}^n X^{(i)}_jX^{(i)}_{n+1-j-i}=X^{(i)}_1X^{(i)}_{n-i}X^{(i)}_2X^{(i)}_{n-1-i}\dotsm X^{(i)}_nX^{(i)}_{1-i}.
\]
Note that $P=P_0P_1\dotsm P_{n-1}$,
since $X^{(i)}_j$ and $X^{(i')}_{j'}$ commute whenever $i\ne i'$ by assumption.

Fix any $i$ with $0\le i\le n-1$.
We define $Y_{2j-1}=X^{(i)}_j$ and $Y_{2j}=X^{(i)}_{n+1-j-i}$.
Note that $P_i=Y_1Y_2\dotsm Y_{2n}$.
By calculation, it can be verified that $Y_{n-i+1-k}=Y_{n-i+k}$ for $1\le k\le n-i$.
Since $Y_j^2=I$ for all $1\le j\le 2n$,
this implies that $Y_1Y_2\dotsm Y_{2(n-i)}=Y_1(Y_2\dotsm(Y_{n-i}Y_{n-i+1})\dotsm Y_{2(n-i)-1})Y_{2(n-i)}=I$.
Similarly, the equation $Y_{2n-i+1-k}=Y_{2n-i+k}$ for $1\le k\le i$
implies that $Y_{2(n-i)+1}\dotsm Y_{2n}=I$.
Therefore $P_i=(Y_1\dotsm Y_{2(n-i)})(Y_{2(n-i)+1}\dotsm Y_{2n})=I$.
This concludes that $P=P_0\dotsm P_{n-1}=I$.

Let $\bra\varphi M_j\ket\varphi=\cos\theta_j$ for $1\le j\le n$,
$\bra\varphi N_k\ket\varphi=\cos\theta'_k$ for $1\le k\le n-1$,
and $-\bra\varphi N_n\ket\varphi=\cos\theta'_n$
with $0\le\theta_j,\theta'_k\le\pi$.
Since $M_1(-N_n)M_2N_{n-1}M_3N_{n-2}\dotsm M_nN_1=-P=-I$,
it holds that $\sum_{j=1}^n\theta_j+\sum_{k=1}^n\theta'_k\ge\pi$
by Corollary~\ref{corollary:2}.
As is shown in the following Lemma~\ref{lemma:max},
$\bra\varphi Z \ket\varphi\le2n\cos(\pi/(2n))$ subject to this constraint,
which establishes Theorem~\ref{theorem:generalized-tsirelson}.
\end{proof}

\begin{lemma}  \label{lemma:max}
  Let $n\ge1$, $0\le\theta_1,\dots,\theta_n\le\pi$ and $\theta_1+\dots+\theta_n\ge\pi$.
  Then $\cos\theta_1+\dots+\cos\theta_n\le n\cos(\pi/n)$.
\end{lemma}

\begin{proof}
  Since the function $\cos\theta$ is decreasing in the range $0\le\theta\le\pi$,
  we may assume that $\theta_1+\dots+\theta_n=\pi$.
  The statement is trivial for $n\le2$.
  We assume $n\ge3$ for the rest of the proof.

  First consider the case where $0\le\theta_1,\dots,\theta_n\le\pi/2$.
  In this case, since the function $\cos\theta$ is concave in the range $0\le\theta\le\pi/2$,
  it follows that $\cos\theta_1+\dots+\cos\theta_n\le n\cos(\pi/n)$.

  Next consider the case where for some $i$, $\theta_i>\pi/2$.
  Without loss of generality, we assume that $\theta_1>\pi/2$.
  Then, again from the concavity of the function $\cos\theta$ in the range $0\le\theta\le\pi/2$,
  it follows that $\cos\theta_2+\dots+\cos\theta_n\le(n-1)\cos((\pi-\theta_1)/(n-1))$.
  Since $\cos\theta_1+(n-1)\cos((\pi-\theta_1)/(n-1))$ is decreasing in the range $\pi/2\le\theta_1\le\pi$,
  \begin{align*}
    \cos\theta_1+\dots+\cos\theta_n &\le \cos\theta_1+(n-1)\cos\frac{\pi-\theta_1}{n-1} \\
    &<\cos\frac{\pi}{2}+(n-1)\cos\frac{\pi}{2(n-1)}<n\cos\frac{\pi}{n}.  \qedhere
  \end{align*}
\end{proof}

To prove Corollary~\ref{corollary:magic-square},
we consider the $n$-player Magic Square game with the assignment $L$
defined as $L=(L_{jk})$ with $L_{jk}\equiv k-j \bmod n$.
We refer to this Latin square as the cyclic Latin square of order $n$,
and this game as the $n$-player Magic Square game with the cyclic assignment.

\begin{proof}[Proof of Corollary~\ref{corollary:magic-square}]
Note that the inequality~(\ref{eq:inequality})
is equivalent to the claim that $\wQ(\MS_n(L))\le(1+\cos\frac{\pi}{2n})/2$
for the cyclic Latin square $L$.
Therefore, Corollary~\ref{corollary:magic-square} follows from
Theorems~\ref{theorem:generalized-tsirelson}, \ref{theorem:ms-assignemt} and \ref{theorem:strategy}.
\end{proof}

We note that Theorem~\ref{theorem:generalized-tsirelson} includes
the following inequality proved by Wehner~\cite{Wehner06} as special cases.

\begin{theorem}[Wehner~\cite{Wehner06}]
  Let $\calH=\calH_1\otimes\calH_2$ be a Hilbert space consisting of two subsystems,
  and let $\ket\varphi\in\calH$ be a state.
  Let $n\ge1$, and let $X_1,\dots,X_n$ be $\pm1$-valued observables on $\calH_1$
  and $Y_1,\dots,Y_n$ be $\pm1$-valued observables on $\calH_2$.
  Then,
  \begin{equation}
    \sum_{j=1}^n \avg{X_jY_j} + \sum_{j=1}^{n-1} \avg{X_{j+1}Y_j} - \avg{X_1Y_n}
    \le 2n\cos\frac{\pi}{2n}.
    \label{eq:wehner}
  \end{equation}
\end{theorem}

\begin{proof}
  In the inequality~(\ref{eq:inequality}),
  let $X^{(0)}_j=I\otimes Y_j$, $X^{(n-1)}_j=X_j\otimes I$, and $X^{(i)}_j=I\otimes I$ for $1\le i\le n-2$.
  Then the inequality~(\ref{eq:inequality}) is exactly the same as the inequality~(\ref{eq:wehner}).
\end{proof}

The equality in (\ref{eq:wehner}) is achievable~\cite{Peres93}.
This gives another proof of $\wq(\MS_n(L))\ge(1+\cos\frac{\pi}{2n})/2$
for the cyclic Latin square $L$
(but not of $\wq(\MS_n)\ge(1+\cos\frac{\pi}{2n})/2$).

\begin{remark}
  For some games $G$, an upper bound on $\wq(G)$
  is obtained from an upper bound on the no-signaling value $\wns(G)$ of $G$,
  which can be characterized by linear programming
  and often easier to compute than $\wq(G)$.
  This is not the case
  for Corollary~\ref{corollary:magic-square} since $\wns(\MS_n)=1$.
  This follows from the result by Barrett and Pironio~%
  \cite[Theorem~1]{BarrettPironio05}:
  for any game $G=(\pi,V)$
  where the predicate $V$ does not depend on the individual answers
  from the players but only on the XOR of all the answers,
  there exists a no-signaling strategy with winning probability one.
\end{remark}

\begin{remark}
We say two Latin squares of order $n$ are equivalent if one is obtained from the other
by swapping rows, swapping columns, relabelling the elements, and/or transposing.
For $n\ge4$, Latin squares of order $n$ is not unique up to this symmetry.
For $n=4$, there are two inequivalent Latin squares:
\[
  L=\begin{array}{|c|c|c|c|} \hline
    1 & 2 & 3 & 4 \\\hline
    4 & 1 & 2 & 3 \\\hline
    3 & 4 & 1 & 2 \\\hline
    2 & 3 & 4 & 1 \\\hline
  \end{array}, \qquad
  L'=\begin{array}{|c|c|c|c|} \hline
    1 & 2 & 3 & 4 \\\hline
    2 & 1 & 4 & 3 \\\hline
    3 & 4 & 1 & 2 \\\hline
    4 & 3 & 2 & 1 \\\hline
  \end{array}.
\]

The first Latin square $L$ is cyclic, but the second Latin square $L'$ is not.
The proof of Corollary~\ref{corollary:magic-square}
depends on the actual assignment of cells to the provers
and it is not applicable to $L'$.
It can be verified by exhaustive search
that for $L'$, the product of the matrices $M_1,M_2,M_3,M_4,N_1,N_2,N_3,N_4$
in any order where each of the eight matrices appears exactly once
is not equal to $-I$ for general matrices $A_{jk}$.
\end{remark}

\section{Three-prover proof system based on three-query PCP}  \label{section:3sat}

\subsection{Construction of proof system}
\label{subsection:3sat-construction}

Let $L\in\naPCP_{1,s(n)}(r(n),3)$.
We construct a three-prover one-round interactive proof system for $L$ as follows.
First, the verifier acts like the PCP verifier
except that, instead of reading the $q_1$th, $q_2$th and $q_3$th bits of the proof,
he writes down the three numbers $q_1,q_2,q_3$.
Next, he performs either \emph{consistency test} or \emph{PCP simulation test}
each with probability $1/2$.
In the consistency test,
the verifier chooses $q\in\{q_1,q_2,q_3\}$ each with probability $1/3$,
and sends $q$ to the three provers.
He accepts if and only if the three answers coincide.
In the PCP simulation test, he sends $q_1,q_2,q_3$ to the three different provers randomly.
He interprets the answers from the provers
as the $q_1$th, $q_2$th and $q_3$th bits in the proof,
and accepts or rejects just as the PCP verifier would do.

This interactive proof system clearly achieves perfect completeness
with honest provers answering the asked bit in the proof. In the
rest of this section, we will show that the soundness acceptance
probability of this interactive proof system with any commuting-operator
provers is at most $1-(1/384)(1-s(n))^2\cdot2^{-2r(n)}$.

Our soundness analysis to prove Theorem~\ref{theorem:pcp-3provers}
shows that for any commuting-operator strategy with high acceptance
probability, there exists a cheating proof string for the underlying
PCP system. The construction of the cheating proof string is similar
to the construction of unentangled strategy used in~\cite{KKMTV-0704.2903v2}.

We note that without the consistency test, the entangled provers can sometimes cheat with certainty.
An example is the well-known GHZ-game,
which corresponds to an unsatisfiable boolean formula
$f=(x_1\xor x_3\xor x_5)\wedge
   (\overline{x_1\xor x_4\xor x_6})\wedge
   (\overline{x_2\xor x_3\xor x_6})\wedge
   (\overline{x_2\xor x_4\xor x_5})$,
where $\xor$ denotes the exclusive OR.

\subsection{Impossibility of perfect cheating}

Before proceeding to the proof of Theorem~\ref{theorem:pcp-3provers},
we first give a much simpler proof of the fact
that entangled or even commuting-operator provers
cannot cheat with certainty in the interactive proof system constructed in the previous subsection
if $x\notin L$.
Such impossibility of perfect cheating was originally proved
in a preliminary work by Sun, Yao and Preda~\cite{SYP07}
with a different proof.
This paper gives a simpler proof of this fact.

Assume that there exists a commuting-operator strategy for perfect
cheating.
We prove that such a strategy essentially satisfies
the condition stated in Lemma~\ref{lemma:commute-classical}.
Precisely speaking, we define a ``good'' subspace $\calH'$ of
$\calH$ containing the shared quantum state such that the
restrictions of the POVM operators to $\calH'$ pairwise commute.

Let $\ket\Psi\in\calH$ be the state shared by the three provers,
and $\calM^{(i)}_q=(M^{(i)}_{q,a})_{a\in\{0,1\}}$ be the PVM measured by prover $i$ for question $q$.
Because the strategy by the provers is accepted with certainty,
it must pass the consistency test in particular.
This means that
$\bra\Psi M^{(i)}_{q,0} M^{(i')}_{q,0} \ket\Psi+\bra\Psi M^{(i)}_{q,1} M^{(i')}_{q,1} \ket\Psi=1$ for $i\ne i'$ and all $q\in Q$,
or equivalently,
\begin{equation}
  M^{(1)}_{q,a}\ket\Psi=M^{(2)}_{q,a}\ket\Psi=M^{(3)}_{q,a}\ket\Psi
  \label{eq:replaceable}
\end{equation}
for all $q\in Q$ and $a\in\{0,1\}$.

Let $\calH'$ be the subspace of $\calH$
spanned by vectors obtained from $\ket\Psi$
by applying zero or more of $M^{(i)}_{q,a}$
for any times and in any order.

\begin{claim} \label{claim:replaceable}
  If $\ket\varphi\in\calH'$, then
  $M^{(1)}_{q,a}\ket\varphi=M^{(2)}_{q,a}\ket\varphi=M^{(3)}_{q,a}\ket\varphi$.
\end{claim}

\begin{proof}
  The proof is by induction on the number $k$ of operators
  applied to $\ket\Psi$ to obtain $\ket\varphi$.

  The case of $k=0$ is by assumption.
  If $k>0$, then $\ket\varphi=M\ket\xi$
  with $M\in\{M^{(1)}_{q',a'},M^{(2)}_{q',a'},M^{(3)}_{q',a'}\}$ for some $q'$ and $a'$, and
  $\ket\xi$ is obtained by applying $M^{(i)}_{q,a}$
  for $k-1$ times to $\ket\Psi$.
  By the induction hypothesis, $\ket\varphi=M^{(1)}_{q',a'}\ket\xi=M^{(2)}_{q',a'}\ket\xi=M^{(3)}_{q',a'}\ket\xi$.
  Therefore, $M^{(1)}_{q,a}\ket\varphi=M^{(2)}_{q,a}\ket\varphi$ since
  $M^{(1)}_{q,a}\ket\varphi=M^{(1)}_{q,a}M^{(3)}_{q',a'}\ket\xi=M^{(3)}_{q',a'}M^{(1)}_{q,a}\ket\xi=M^{(3)}_{q',a'}M^{(2)}_{q,a}\ket\xi
   =M^{(2)}_{q,a}M^{(3)}_{q',a'}\ket\xi=M^{(2)}_{q,a}\ket\varphi$, here we use the fact that $M^{(i)}_{q,a}$ and $M^{(i')}_{q',a'}$ commute whenever $i\ne i'$.
  The equation $M^{(2)}_{q,a}\ket\varphi=M^{(3)}_{q,a}\ket\varphi$ is proved similarly.
\end{proof}

\begin{claim} \label{claim:commutative}
  The $6n$ projectors $M^{(i)}_{q,a}$ pairwise commute on $\calH'$.
\end{claim}

\begin{proof}
  Let $\ket\varphi\in\calH'$.
  By Claim~\ref{claim:replaceable},
  $M^{(1)}_{q,a}M^{(1)}_{q',a'}\ket\varphi
   =M^{(1)}_{q,a}M^{(3)}_{q',a'}\ket\varphi
   =M^{(3)}_{q',a'}M^{(1)}_{q,a}\ket\varphi
   =M^{(3)}_{q',a'}M^{(2)}_{q,a}\ket\varphi
   =M^{(2)}_{q,a}M^{(3)}_{q',a'}\ket\varphi
   =M^{(2)}_{q,a}M^{(1)}_{q',a'}\ket\varphi
   =M^{(1)}_{q',a'}M^{(2)}_{q,a}\ket\varphi
   =M^{(1)}_{q',a'}M^{(1)}_{q,a}\ket\varphi$.
  The equations $M^{(2)}_{q,a}M^{(2)}_{q',a'}\ket\varphi=M^{(2)}_{q',a'}M^{(2)}_{q,a}\ket\varphi$
  and $M^{(3)}_{q,a}M^{(3)}_{q',a'}\ket\varphi=M^{(3)}_{q',a'}M^{(3)}_{q,a}\ket\varphi$
  are proved similarly.
\end{proof}

Note that $\ket\Psi\in\calH'$
and that $\calH'$ is invariant under each $M^{(i)}_{q,a}$.
This means that we could use $\calH'$ instead of $\calH$ in the first place.
By Claim~\ref{claim:commutative}, these $6n$ operators are
pairwise commuting Hermitian operators when restricted to $\calH'$.
By Lemma~\ref{lemma:commute-classical}, there exists a classical strategy
achieving the same acceptance probability $1$,
and therefore the original PCP is accepted with certainty.
This means that if $x\notin L$, the commuting-operator provers cannot achieve perfect cheating.

\begin{remark}
  A statement analogous to Claim~\ref{claim:commutative}
  does not hold if there are only two provers.
  For example,
  let $\ket\Psi=(\lvert01\rangle-\lvert10\rangle)/\sqrt2\in\CC^2\otimes\CC^2$.
  Let $M_1,M_2$ be arbitrary Hermitian projectors on $\CC^2$
  such that $M_1$ and $M_2$ do not commute,
  and let $M^{(1)}_{q,1}=M_q\otimes I$, $M^{(2)}_{q,1}=I\otimes(I-M_q)$ for $q=1,2$.
  Then $M^{(1)}_{q,a}\ket\Psi=M^{(2)}_{q,a}\ket\Psi$ for $q\in\{1,2\}$ and $a\in\{0,1\}$
  whereas $M^{(1)}_{1,a}M^{(1)}_{2,a'}\ket\Psi\ne M^{(1)}_{2,a'}M^{(1)}_{1,a}\ket\Psi$.
\end{remark}

\subsection{Proof of Theorem~\ref{theorem:pcp-3provers}}

In the case of imperfect cheating,
the equalities in (\ref{eq:replaceable}) hold only approximately,
and we cannot define a ``good'' subspace $\calH'$ as in the case of perfect cheating.
Instead, we will prove that an approximate version of the equation (\ref{eq:replaceable})
implies that measurements $\calM^{(i)}_q$ are almost commuting on the shared state $\ket\Psi$.

Kempe, Kobayashi, Matsumoto, Toner and Vidick~\cite{KKMTV-0704.2903v2} prove
soundness of their classical three-prover interactive proof system
by comparing the behavior of the first and second provers in an
arbitrary entangled strategy to that in the strategy modified as
follows: instead of measuring the answer to the asked question, the
two provers always measure the answers to all possible questions and
just send back the answer to the asked question. This modification
makes the behavior classical. The key in their proof is that if the
third prover answers consistently with high probability,
the measurements performed by the first and second provers do not
disturb the reduced state shared by them so much (Claim~20
in~\cite{KKMTV-0704.2903v2}), and the modification above does not decrease the
acceptance probability so much.

We will use a similar idea when constructing a proof string for the original PCP system,
but instead of the non-disturbance property, we use the fact that all the POVMs almost commute on $\ket\Psi$.
This modification of the proof technique seems necessary
because taking partial trace is meaningless in the commuting-operator model.

The following lemma is the key to bound the difference between two POVMs applied to states other than $\ket\Psi$.

\begin{lemma}  \label{lemma:filter}
  Let $\rho$ be a density matrix,
  and $\calM=(M_i)_{i=1}^v$ and $\calN=(N_i)_{i=1}^v$ be POVMs.
  Let
  \begin{align*}
    \lambda&=\frac12\sum_{i=1}^v\tr\rho(\sqrt{M_i}-\sqrt{N_i})^2
    =1-\sum_{i=1}^v \tr\rho \frac{\sqrt{M_i}\sqrt{N_i}+\sqrt{N_i}\sqrt{M_i}}{2}, \\
    \Delta&=\sum_{i=1}^v \norm[\big]{{\sqrt{M_i}\rho\sqrt{M_i}-\sqrt{N_i}\rho\sqrt{N_i}}}_{\tr}.
  \end{align*}
  Then $\Delta\le2\sqrt{2\lambda}$.
\end{lemma}

\begin{proof}
  Let $X_i=\sqrt{M_i}$ and $Y_i=\sqrt{N_i}$.
  First we prove the case where $\rho$ is a pure state: $\rho=\ket\Psi\bra\Psi$.
  We define vectors $\vct{x},\vct{y},\vct{z}\in\RR^v$ by
  $x_i=\norm{X_i\ket\Psi}$,
  $y_i=\norm{Y_i\ket\Psi}$, and
  $z_i=\norm{(X_i-Y_i)\ket\Psi}$.
  By using these vectors, $\Delta$ can be bounded as $\Delta\le(\vct{x}+\vct{y})\cdot\vct{z}$, since
  \begin{align*}
    \Delta
    &=\sum_{i=1}^v \Big\| X_i\ket\Psi\bra\Psi X_i-X_i\ket\Psi\bra\Psi Y_i+X_i\ket\Psi\bra\Psi Y_i-Y_i\ket\Psi\bra\Psi Y_i\Big\|_{\tr} \\
    &\le\sum_{i=1}^v \left(\Big\|X_i\ket\Psi\bra\Psi X_i-X_i\ket\Psi\bra\Psi Y_i\Big\|_{\tr}+\Big\|X_i\ket\Psi\bra\Psi Y_i-Y_i\ket\Psi\bra\Psi Y_i\Big\|_{\tr}\right) \\
    &=\sum_{i=1}^v \left(\Big\|X_i\ket\Psi\langle\Psi|\Psi\rangle\bra\Psi(X_i-Y_i)\Big\|_{\tr}+\Big\|(X_i-Y_i)\ket\Psi\langle\Psi|\Psi\rangle\bra\Psi Y_i\Big\|_{\tr}\right) \\
    &\le\sum_{i=1}^v \left(\Big\|X_i\ket\Psi\bra\Psi\big\|+\Big\|Y_i\ket\Psi\bra\Psi\Big\|\right)\Big\|(X_i-Y_i)\ket\Psi\bra\Psi\Big\|_{\tr} \\
    &=\sum_{i=1}^v \Big(\|X_i\ket\Psi\|+\|Y_i\ket\Psi\|\Big)\Big\|(X_i-Y_i)\ket\Psi\Big\| \\
    &=(\vct{x}+\vct{y})\cdot\vct{z}.
  \end{align*}
  Note that $\vct{x}$ is a unit vector since
  \[
    \norm{\vct{x}}^2
    =\sum_{i=1}^v \norm{X_i\ket\Psi}^2
    =\bra\Psi \Bigl(\sum_{i=1}^v M_i\Bigr) \ket\Psi
    =\norm{\ket\Psi}^2=1,
  \]
  and similarly $\norm{\vct{y}}^2=1$.
  Moreover,
  \[
    \norm{\vct{z}}^2=\sum_{i=1}^v \norm{(X_i-Y_i)\ket\Psi}^2=2\lambda.
  \]
  Therefore, $\Delta\le(\vct{x}+\vct{y})\cdot\vct{z}\le\norm{\vct{x}+\vct{y}}\,\norm{\vct{z}}\le2\sqrt{2\lambda}$.

  If $\rho$ is a mixed state, decompose $\rho$ to a convex combination of pure states:
  $\rho=\sum_{j=1}^n p_j\rho_j$.
  Let
  \begin{align*}
    \lambda_j&=\frac12\sum_{i=1}^v\tr\rho_j(X_i-Y_i)^2, \\
    \Delta_j&=\sum_{i=1}^v\norm{X_i\rho_j X_i-Y_i\rho_j Y_i}.
  \end{align*}
  Then,
  \[
    \Delta\le\sum_{j=1}^n p_j\Delta_j\le\sum_{j=1}^n p_j\cdot2\sqrt{2\lambda_j}\le2\sqrt{2\lambda}.  \qedhere
  \]
\end{proof}

We fix an input $x\notin L$.
Let $Q\subseteq\ZZ_{>0}$ be the set of indices of the bits in a proof string
which are queried by the PCP verifier with nonzero probability,
and $N$ be the maximum of the elements of $Q$.
Note that $\abs{Q}\le3\cdot 2^r$.
Let $\pi(q_1,q_2,q_3)$ be the probability
with which the PCP verifier reads the $q_1$th, $q_2$th and $q_3$th bits in the proof at the same time
($\sum_{q_1,q_2,q_3\in Q} \pi(q_1,q_2,q_3)=1$).
Without loss of generality, we assume that $\pi(q_1,q_2,q_3)$ is symmetric
and that $\pi(q_1,q_2,q_3)=0$ if $q_1,q_2,q_3$ are not all distinct.
For $q_1,q_2,q_3\in Q$ and $a_1,a_2,a_3\in\{0,1\}$,
let $V(a_1,a_2,a_3\mid q_1,q_2,q_3)=1$ if the PCP verifier accepts
when he asks the $q_1$th, $q_2$th and $q_3$th bits in the proof
and receives the corresponding answers $a_1$, $a_2$ and $a_3$,
and $V(a_1,a_2,a_3\mid q_1,q_2,q_3)=0$ otherwise.
For $q\in Q$, let $\pi_q=\sum_{q_2,q_3\in Q} \pi(q,q_2,q_3)=\sum_{q_1,q_3\in Q} \pi(q_1,q,q_3)=\sum_{q_1,q_2\in Q} \pi(q_1,q_2,q)$.
For simplicity, we let $\pi_q=0$ for $q\notin Q$.

Consider an arbitrary commuting-operator strategy
for the constructed three-prover one-round interactive proof system,
and let $w$ be its acceptance probability.
By Lemma~\ref{lemma:symmetrize},
we can assume that this strategy is symmetric without loss of generality.
Let $\ket\Psi$ be the quantum state shared by the provers.
For $1\le i\le3$ and $q\in Q$,
let $\calM^{(i)}_q=(M^{(i)}_{q,0},M^{(i)}_{q,1})$ be the PVM
measured by the $i$th prover when asked the $q$th bit in the proof.
For simplicity, we let $M^{(i)}_{q,0}=I$ and $M^{(i)}_{q,1}=0$ for $q\notin Q$.
Then, when asked the $q_1$th, $q_2$th and $q_3$th bits in the proof, the provers answer $a_1,a_2,a_3\in\{0,1\}$
with probability
\[
  \PQ(a_1,a_2,a_3\mid q_1,q_2,q_3)=\left\|M^{(1)}_{q_1,a_1}M^{(2)}_{q_2,a_2}M^{(3)}_{q_3,a_3}\ket\Psi\right\|^2.
\]
Because the strategy is symmetric,
it holds that $\bra\Psi M^{(1)}_{q,a}M^{(2)}_{q,a}\ket\Psi=\bra\Psi M^{(2)}_{q,a}M^{(3)}_{q,a}\ket\Psi=\bra\Psi M^{(3)}_{q,a}M^{(1)}_{q,a}\ket\Psi$.
Let
\begin{align*}
  \lambda_q&=1-\sum_{a\in\{0,1\}}\bra\Psi M^{(1)}_{q,a}M^{(2)}_{q,a}\ket\Psi \\
           &=1-\sum_{a\in\{0,1\}}\bra\Psi M^{(2)}_{q,a}M^{(3)}_{q,a}\ket\Psi \\
           &=1-\sum_{a\in\{0,1\}}\bra\Psi M^{(3)}_{q,a}M^{(1)}_{q,a}\ket\Psi.
\end{align*}
Note that $\lambda_q=0$ for $q\notin Q$.
Now we can write $w$ as $w=(\wcons+\wproof)/2$, where
\begin{align*}
  \wcons
  &=\sum_{q\in Q} \pi_q \Big(\PQ(0,0,0\mid q,q,q)+\PQ(1,1,1\mid q,q,q)\Big) \\
  &=\sum_{q\in Q} \pi_q \left(\bra\Psi M^{(1)}_{q,0}M^{(2)}_{q,0}M^{(3)}_{q,0}\ket\Psi+\bra\Psi M^{(1)}_{q,1}M^{(2)}_{q,1}M^{(3)}_{q,1}\ket\Psi\right) \\
  &=\sum_{q\in Q} \pi_q \frac{\sum_{a\in\{0,1\}}(
      \bra\Psi M^{(1)}_{q,a}M^{(2)}_{q,a}\ket\Psi
     +\bra\Psi M^{(2)}_{q,a}M^{(3)}_{q,a}\ket\Psi
     +\bra\Psi M^{(3)}_{q,a}M^{(1)}_{q,a}\ket\Psi)-1}{2}
   =1-\frac32\sum_{q\in Q} \pi_q \lambda_q, \\
  \wproof
  &=\sum_{q_1,q_2,q_3\in Q} \pi(q_1,q_2,q_3) \sum_{a_1,a_2,a_3\in\{0,1\}} \PQ(a_1,a_2,a_3\mid q_1,q_2,q_3)V(a_1,a_2,a_3\mid q_1,q_2,q_3).
\end{align*}
Since $\pi_q\ge1/(3\cdot2^r)$ for all $q\in Q$, we have
\begin{equation}
  \wcons\le 1-\frac{1}{2\cdot2^r}\sum_{q\in Q} \lambda_q.
  \label{eq:wcons}
\end{equation}

We construct a random proof string $y=y_1\dotsm y_N$ according to the probability distribution
\[
  \Pr(y_1,\dots,y_N)
  =\left\|M^{(i)}_{N,y_N}\dotsm M^{(i)}_{1,y_1}\ket\Psi\right\|^2.
\]
Note that the value of the right-hand side does not depend on the
choice of $i$ because of the symmetry. For distinct $q_1,q_2,q_3\in
Q$ and for $a_1,a_2,a_3\in\{0,1\}$, the joint probability of the events
$y_{q_1}=a_1$, $y_{q_2}=a_2$, $y_{q_3}=a_3$ is given by
\[
  \Pc(a_1,a_2,a_3\mid q_1,q_2,q_3)=\sum_{\substack{y\in\{0,1\}^N \\ y_{q_1}=a_1,y_{q_2}=a_2,y_{q_3}=a_3}}\Pr(y_1,\dots,y_N).
\]
By the soundness condition of the PCP system,
\[
  \sum_{q_1,q_2,q_3\in Q} \pi(q_1,q_2,q_3) \sum_{a_1,a_2,a_3\in\{0,1\}} \Pc(a_1,a_2,a_3\mid q_1,q_2,q_3)V(a_1,a_2,a_3\mid q_1,q_2,q_3) \le s.
\]

We will prove that if $\wcons$ is large,
then the difference between $\PQ$ and $\Pc$ is not large
and therefore $\wproof$ is not much larger than $s$.

For $a_1,a_2,a_3\in\{0,1\}$ and distinct $q_1,q_2,q_3\in Q$, let
\[
  P'(a_1,a_2,a_3\mid q_1,q_2,q_3)
  =\norm{M^{(i)}_{q'_1,a'_1}M^{(i)}_{q'_2,a'_2}M^{(i)}_{q'_3,a'_3}\ket\Psi}^2,
\]
where
$\{(a'_1,q'_1),\allowbreak(a'_2,q'_2),\allowbreak(a'_3,q'_3)\}
 \allowbreak=\allowbreak
 \{(a_1,q_1),\allowbreak(a_2,q_2),\allowbreak(a_3,q_3)\}$
and $q'_1<q'_2<q'_3$.
Again the value of the right-hand side does not depend on the choice of $i$.

\setcounter{claim}{0}
\begin{claim}  \label{claim:1}
  For distinct $q_1,q_2,q_3\in Q$,
  \[
    \sum_{a_1,a_2,a_3\in\{0,1\}}\abs{\Pc(a_1,a_2,a_3\mid q_1,q_2,q_3)-P'(a_1,a_2,a_3\mid q_1,q_2,q_3)}
    \le\sum_{q=1}^{\max\{q_1,q_2,q_3\}} 2\sqrt{2\lambda_q}.
  \]
\end{claim}

\begin{proof}
  We may assume without loss of generality that $1\le q_1<q_2<q_3\le N$.
  Let $l=q_3$.
  We prove the claim by hybrid argument.
  To do this, we shall define probability distributions $p_0,\dots,p_l$ on $\{0,1\}^l$
  such that $p_0$ and $p_l$ are related to $\Pc$ and $P'$, respectively.
  For $1\le q\le l$, we define $i_q$ as
  $i_q=1$ if $q\in\{q_1,q_2,q_3\}$ and $i_q=2$ otherwise.
  Note that $M^{(i_q)}_{q,a}$ commutes with $M^{(3)}_{q',a'}$ for all $1\le q'\le l$ and $a'\in\{0,1\}$ in either case.%
  \footnote{This argument is the reason why we need three provers.}
  For $0\le q\le l$ and $y\in\{0,1\}^l$, let
  \[
    p_q(y)=\norm{M^{(i_1)}_{1,y_1}M^{(i_2)}_{2,y_2}\dotsm M^{(i_q)}_{q,y_q}M^{(3)}_{l,y_l}M^{(3)}_{l-1,y_{l-1}}\dotsm M^{(3)}_{q+1,y_{q+1}}\ket\Psi}^2.
  \]
  For $a_1,a_2,a_3\in\{0,1\}$,
  \begin{align*}
    \sum_{\substack{y\in\{0,1\}^l \\ y_{q_1}=a_1,y_{q_2}=a_2,y_{q_3}=a_3}}p_0(y) &= \Pc(a_1,a_2,a_3\mid q_1,q_2,q_3), \\
    \sum_{\substack{y\in\{0,1\}^l \\ y_{q_1}=a_1,y_{q_2}=a_2,y_{q_3}=a_3}}p_l(y)
    &= \norm{M^{(1)}_{q_1,a_1}M^{(1)}_{q_2,a_2}M^{(1)}_{q_3,a_3}\ket\Psi}^2
     = P'(a_1,a_2,a_3\mid q_1,q_2,q_3).
  \end{align*}
  Let $1\le q\le l$.
  By Lemma~\ref{lemma:filter}, we have
  \[
    \sum_{y_q\in\{0,1\}}\norm[\big]{{
     M^{(3)}_{q,y_q}\ket\Psi\bra\Psi M^{(3)}_{q,y_q}
    -M^{(i_q)}_{q,y_q}\ket\Psi\bra\Psi M^{(i_q)}_{q,y_q}
    }}_{\tr}\le2\sqrt{2\lambda_q}.
  \]
  Since the trace distance between two states is an upper bound on the statistical difference
  between the probability distributions resulting from making the same measurement on the two states,
  \[
    \sum_{y\in\{0,1\}^l}\abs[\big]{{
     \norm{M^{(i_1)}_{1,y_1}\dotsm M^{(i_{q-1})}_{q-1,y_{q-1}}M^{(3)}_{l,y_l}\dotsm M^{(3)}_{q+1,y_{q+1}}M^{(3)}_{q,y_q}\ket\Psi}^2
    -\norm{M^{(i_1)}_{1,y_1}\dotsm M^{(i_{q-1})}_{q-1,y_{q-1}}M^{(3)}_{l,y_l}\dotsm M^{(3)}_{q+1,y_{q+1}}M^{(i_q)}_{q,y_q}\ket\Psi}^2
    }}\le2\sqrt{2\lambda_q},
  \]
  or equivalently,
  \[
    \sum_{y\in\{0,1\}^l}\abs{p_{q-1}(y)-p_q(y)}\le2\sqrt{2\lambda_q}.
  \]
  Summing up this inequality for $1\le q\le l$, we obtain
  \[
    \sum_{y\in\{0,1\}^l}\abs{p_0(y)-p_l(y)}\le\sum_{q=1}^l 2\sqrt{2\lambda_q}
  \]
  by the triangle inequality, or equivalently,
  \[
    \sum_{a_1,a_2,a_3\in\{0,1\}}\sum_{\substack{y\in\{0,1\}^l \\ y_{q_1}=a_1,y_{q_2}=a_2,y_{q_3}=a_3}}\abs{p_0(y)-p_l(y)}\le\sum_{q=1}^l 2\sqrt{2\lambda_q}.
  \]
  The claim follows by moving the summation over $y$ inside the
  absolute value by using the triangle inequality.
\end{proof}

\begin{claim}  \label{claim:2}
  For distinct $q_1,q_2,q_3\in Q$,%
  \footnote{Actually, we can omit the term $2\sqrt{2\lambda_{q_i}}$
    from the right-hand side of the inequality,
    where $q_i=\min\{q_1,q_2,q_3\}$.}
  \[
    \sum_{a_1,a_2,a_3\in\{0,1\}}\abs{P'(a_1,a_2,a_3\mid q_1,q_2,q_3)-\PQ(a_1,a_2,a_3\mid q_1,q_2,q_3)}
    \le2\sqrt{2\lambda_{q_1}}+2\sqrt{2\lambda_{q_2}}+2\sqrt{2\lambda_{q_3}}.
  \]
\end{claim}

\begin{proof}
  If $q_1<q_2<q_3$, sum up the two inequalities
  \begin{align*}
    \sum_{a_1,a_2,a_3\in\{0,1\}}
    \abs[\big]{{ \norm{M^{(1)}_{q_1,a_1}M^{(1)}_{q_2,a_2}M^{(1)}_{q_3,a_3}\ket\Psi}^2
                -\norm{M^{(1)}_{q_1,a_1}M^{(1)}_{q_2,a_2}M^{(3)}_{q_3,a_3}\ket\Psi}^2}}
    &\le2\sqrt{2\lambda_{q_3}}, \\
    \sum_{a_1,a_2,a_3\in\{0,1\}}
    \abs[\big]{{ \norm{M^{(3)}_{q_3,a_3}M^{(1)}_{q_1,a_1}M^{(1)}_{q_2,a_2}\ket\Psi}^2
                -\norm{M^{(3)}_{q_3,a_3}M^{(1)}_{q_1,a_1}M^{(2)}_{q_2,a_2}\ket\Psi}^2}}
    &\le2\sqrt{2\lambda_{q_2}},
  \end{align*}
  each of which follows from Lemma~\ref{lemma:filter},
  and use the triangle inequality.
  The other cases are proved similarly,
  where we use $P'(a_1,a_2,a_3\mid q_1,q_2,q_3)=\norm{M^{(i)}_{q_1,a_1}M^{(i)}_{q_2,a_2}M^{(i)}_{q_3,a_3}\ket\Psi}^2$
  with $i$ such that $q_i$ is the smallest in $q_1,q_2,q_3$.
\end{proof}

By Claims~\ref{claim:1} and \ref{claim:2}, for any distinct
$q_1,q_2,q_3\in Q$,
\begin{align*}
  &\sum_{a_1,a_2,a_3\in\{0,1\}}\abs{\Pc(a_1,a_2,a_3\mid q_1,q_2,q_3)-\PQ(a_1,a_2,a_3\mid q_1,q_2,q_3)} \\
  \le\;&2\sqrt{2\lambda_{q_1}}+2\sqrt{2\lambda_{q_2}}+2\sqrt{2\lambda_{q_3}}+\sum_{q=1}^{\max\{q_1,q_2,q_3\}} 2\sqrt{2\lambda_q} \\
  \le\;&4\sqrt2 \sum_{q\in Q} \sqrt{\lambda_q}.
\end{align*}
Therefore,
\[
  \abs{\wproof-s}
  \le 4\sqrt2\sum_{q\in Q}\sqrt{\lambda_q}
  \le 4\sqrt2\sqrt{\abs{Q}\sum_{q\in Q}\lambda_q}
  \le 4\sqrt2\sqrt{2\cdot2^r\abs{Q}(1-\wcons)}
  \le8\sqrt3\cdot2^r\sqrt{1-\wcons},
\]
where the third inequality follows from the inequality~(\ref{eq:wcons})
and the last inequality follows from the fact $\abs{Q}\le3\cdot2^r$.
This implies%
\footnote{The first inequality is shown as follows.
  Let $c=8\sqrt3\cdot2^r$, $t=1-\wproof$, $u=\sqrt{1-\wcons}$.
  Then $8\sqrt6\cdot2^r\sqrt{1-w}=c\sqrt{t+u^2}$.
  Since $c\ge8\sqrt3$,
  it follows that $c^2(t+u^2)-(t+cu)^2=c^2t-t^2-2ctu\ge t(c^2-t-2c)\ge t(c^2-1-2c)\ge0$,
  or $c\sqrt{t+u^2}\ge t+cu=1-\wproof+8\sqrt3\cdot2^r\sqrt{1-\wcons}$.}
\[
  8\sqrt6\cdot2^r\sqrt{1-w}=8\sqrt3\cdot2^r\sqrt{(1-\wproof)+(1-\wcons)}\ge1-\wproof+8\sqrt3\cdot2^r\sqrt{1-\wcons}\ge1-s,
\]
or equivalently $1-w\ge(1/384)(1-s)^2\cdot2^{-2r}$.

\subsection{The two-prover case}

Finally, the result by Cleve, H\o yer, Toner and Watrous~\cite{CHTW04}
essentially implies that it is efficiently decidable
whether the entangled value of a given two-player one-round binary-answer game is equal to one or not.
This proves Theorem~\ref{theorem:2provers}.

\begin{proof}[Proof of Theorem~\ref{theorem:2provers}]
  \begin{enumerate}[(i)]
  \item
    For a two-player one-round binary-answer game $G$,
    $\wq(G)=1$ if and only if $\wc(G)=1$~\cite[Theorem~5.12]{CHTW04}.
    Therefore, the problem of deciding whether $\wq(G)=1$ or not
    is equivalent to a problem of deciding whether $\wc(G)=1$ or not.
    Since $G$ is two-player and binary-answer,
    testing whether $\wc(G)=1$ or not can be cast as an instance
    of the 2SAT problem,
    and it is solvable in time polynomial in the number of questions.
  \item
    This part follows from (i)
    since any classical two-prover one-round binary interactive proof system
    with entangled provers
    involves at most exponentially many questions.
    \qedhere
  \end{enumerate}
\end{proof}

\section*{Acknowledgement}

The authors are grateful to Keiji Matsumoto and anonymous reviewers
for their helpful comments.


\begin{thebibliography}{10}

\bibitem{BFL91}
L{\'{a}}szl{\'{o}} Babai, Lance Fortnow, and Carsten Lund.
\newblock Non-deterministic exponential time has two-prover interactive
  protocols.
\newblock {\em Computational Complexity}, 1(1):3--40, 1991.

\bibitem{BarrettPironio05}
Jonathan Barrett and Stefano Pironio.
\newblock {Popescu}--{Rohrlich} correlations as a unit of nonlocality.
\newblock {\em Physical Review Letters}, 95(140401), September 2005.

\bibitem{BGKW88}
Michael Ben-Or, Shafi Goldwasser, Joe Kilian, and Avi Wigderson.
\newblock Multi-prover interactive proofs: How to remove intractability
  assumptions.
\newblock In {\em Proceedings of the Twentieth Annual ACM Symposium on Theory
  of Computing}, pages 113--131, 1988.

\bibitem{CHSH69}
John~F. Clauser, Michael~A. Horne, Abner Shimony, and Richard~A. Holt.
\newblock Proposed experiment to test local hidden-variable theories.
\newblock {\em Physical Review Letters}, 23(15):880--884, 1969.

\bibitem{CGJ07}
Richard Cleve, Dmitry Gavinsky, and Rahul Jain.
\newblock Entanglement-resistant two-prover interactive proof systems and
  non-adaptive private information retrieval systems.
\newblock {arXiv}:0707.1729v1 [quant-ph], 2007.

\bibitem{CHTW04}
Richard Cleve, Peter H{\o}yer, Benjamin Toner, and John Watrous.
\newblock Consequences and limits of nonlocal strategies.
\newblock In {\em Proceedings of the Nineteenth IEEE Annual Conference on
  Computational Complexity}, pages 236--249, 2004.

\bibitem{DK00}
Ding-Zhu Du and Ker-I Ko.
\newblock {\em Theory of Computational Complexity}.
\newblock Series in Discrete Mathematics and Optimization. Wiley-Interscience,
  2000.

\bibitem{FL92}
Uriel Feige and L{\'{a}}szl{\'{o}} Lov{\'{a}}sz.
\newblock Two-prover one-round proof systems: Their power and their problems.
\newblock In {\em Proceedings of the Twenty-Fourth Annual ACM Symposium on
  Theory of Computing}, pages 733--744, 1992.

\bibitem{FRS94}
Lance Fortnow, John Rompel, and Michael Sipser.
\newblock On the power of multi-prover interactive protocols.
\newblock {\em Theoretical Computer Science}, 134(2):545--557, 1994.

\bibitem{Hastad01}
Johan H{\aa}stad.
\newblock Some optimal inapproximability results.
\newblock {\em Journal of the ACM}, 48(4):798--859, 2001.

\bibitem{KKMTV-0704.2903v2}
Julia Kempe, Hirotada Kobayashi, Keiji Matsumoto, Ben Toner, and Thomas Vidick.
\newblock Entangled games are hard to approximate.
\newblock {arXiv}:0704.2903v2 [quant-ph], 2007.

\bibitem{KM03}
Hirotada Kobayashi and Keiji Matsumoto.
\newblock Quantum multi-prover interactive proof systems with limited prior
  entanglement.
\newblock {\em Journal of Computer and System Sciences}, 66(3):429--450, 2003.

\bibitem{NC01}
Michael~A. Nielsen and Isaac~L. Chuang.
\newblock {\em Quantum Computation and Quantum Information}.
\newblock Cambridge University Press, 2001.

\bibitem{Peres93}
Asher Peres.
\newblock {\em Quantum Theory: Concepts and Methods}, volume~57 of {\em
  Fundamental Theories of Physics}.
\newblock Kluwer Academic Publishers, 1993.

\bibitem{SYP07}
Xiaoming Sun, Andrew C.-C.~Yao, and Daniel Preda.
\newblock On entangled quantum 3-prover systems for {SAT} and the magic square.
\newblock Invited talk at QIP 2007 presented by A.~Yao, 2007.

\bibitem{Tsirelson80}
Boris~S. Tsirelson.
\newblock Quantum generalizations of {Bell's} inequality.
\newblock {\em Letters in Mathematical Physics}, 4(2):93--100, 1980.

\bibitem{Tsirelson06}
Boris~S. Tsirelson.
\newblock {Bell} inequalities and operator algebras.
\newblock \url|http://www.tau.ac.il/~tsirel/download/bellopalg.html|, 2006.

\bibitem{Wehner06}
Stephanie Wehner.
\newblock {Tsirelson} bounds for generalized {Clauser--Horne--Shimony--Holt}
  inequalities.
\newblock {\em Physical Review A}, 73(022110), 2006.

\end{thebibliography}
\end{document}